\documentclass[11pt]{article}

\usepackage[margin=1in]{geometry}
\usepackage{times}
\usepackage[utf8]{inputenc}
\usepackage[T1]{fontenc}
\usepackage[english]{babel}
\usepackage{natbib}
\setcitestyle{authoryear,round,citesep={;},aysep={,},yysep={;}}
\usepackage{amsmath}
\usepackage{amssymb}
\usepackage{amsfonts}
\usepackage{booktabs}
\usepackage{etoolbox}
\usepackage{microtype}
\usepackage{xcolor}
\usepackage{array}
\usepackage{graphicx}
\usepackage{tikz}
\usetikzlibrary{positioning,fit,calc,backgrounds,shapes.symbols,arrows.meta}
\usepackage[hidelinks]{hyperref}
\usepackage{cleveref}
\usepackage{amsthm}

\theoremstyle{plain}
\newtheorem{theorem}{Theorem}[section]

\theoremstyle{definition}

\theoremstyle{remark}

\newenvironment{promptbox}{\list{}{\leftmargin=1em\rightmargin=0pt}\item\relax\fontsize{8}{9.5}\selectfont}{\endlist}

\graphicspath{{figures/}}

\AtBeginEnvironment{table}{\setlength{\belowcaptionskip}{6pt}}


\title{\texorpdfstring{Learning Checklists to Reduce Lottery Effects\\in Scientific Peer Review}{Learning Checklists to Reduce Lottery Effects in Scientific Peer Review}}

\author{%
  \begin{tabular}[t]{@{}c@{\hspace{3.5em}}c@{}}
    Xuanyou Liu & Ziyang Guo \\
    {\small\texttt{xuanyou@u.northwestern.edu}} & {\small\texttt{ziyang.guo@northwestern.edu}} \\[0.9em]
    Jason Hartline & Jessica Hullman \\
    {\small\texttt{hartline@northwestern.edu}} & {\small\texttt{jhullman@northwestern.edu}} \\[0.9em]
    \multicolumn{2}{c}{Department of Computer Science, Northwestern University}
  \end{tabular}%
}
\date{}

\begin{document}

\maketitle

\begin{abstract}
\noindent
Reviewer lottery error describes error due to a small set of peer reviews attending to only some of the issues with a paper that a larger pool of eligible reviewers would raise. We show this error can be reduced by surfacing a checklist of criticisms learned from a historical review corpus. 
Our approach generates a tree of checklists representing coarsenings of the atomic review claim space, then selects the coarsening that maximizes a theoretically-motivated measure of correlated agreement among reviewers. We realize this approach with an LLM-assisted pipeline and test it on NeurIPS~2021 and ICLR~2022 review data. In simulated review with LLM reviewers, the learned checklist reduces the mean squared error (MSE) of a three-reviewer committee by 58\% on NeurIPS and 21\% on ICLR, and on both venues the checklist with the highest correlated agreement has the lowest error for committees of three or more. In a human study with 20 participants, the learned NeurIPS checklist lowers the MSE of a paper ranking, measured against the NeurIPS~2021 order, from 0.48 to 0.22. 

\end{abstract}

\section{Introduction}
Peer review systems have long played the role of a quality control mechanism for scientific production. However, human peer review is far from being a perfect system, as a result of bias, poor incentives, limited expertise, and other issues~\citep{mann2025ai,shah2022challenges}.
One recognized problem is reviewer lottery variance: a different assignment of qualified reviewers may have resulted in a different outcome. ``Luck of the draw'' effects have been demonstrated across various fields~\citep{cole1981chance,mayo2006peering,beygelzimer2023has,cortes2021inconsistency}. For example, the paired-committee 2021 NeurIPS reviewing experiment found that the two committees disagreed on 23\% of papers, and that roughly half of the accepted papers would have been replaced if review assignments were rerun~\citep{beygelzimer2023has}. Beyond differences in expertise and in how reviewers perceive the paper, variance can arise because any single reviewer covers only a small portion of the evaluation space~\citep{deng2026cocoreviewbench}.

More severely, as LLMs increase scientific production~\citep{kusumegi2025scientific}, higher reviewing loads per human reviewer increase noise, instigating a ``review death spiral'' that further incentivizes authors to try their luck by submitting weak papers~\citep{bergstrom2026screening,liu2026aireview}.

LLMs have been proposed as part of the solution~\citep{liang2024large,xu2025limitations,wei2025aiimperative}. For example, by systematically conducting and surfacing a broad set of checks across papers, AI reviewers can reduce lottery variance by compensating for the limited attention and idiosyncrasies of individual human reviewers. This ensures a common set of baseline checks is applied to each paper, while still allowing the individual reviewers to exercise their judgment and expertise in weighing surfaced issues and identifying other aspects of the paper to assess. However, the risk that must be balanced is that a shared representation evaluating the paper can introduce bias by causing reviewers to underweight their private information~\citep[see, e.g.,][]{bikhchandani1992theory,stasser1985pooling}.

We first characterize the problem of designing an AI checklist as identifying a representation of the review claim space that, when applied to a paper, reproducibly reduces lottery variance without increasing overall error via the addition of shared bias. We formalize checklist design as learning a coarsened representation of the review claim space observed in historical reviews. Because historical peer review datasets do not enable estimating reviewer decisions under checklists that reviewers didn't observe, our approach targets the coarsening that maximizes a proxy objective based in correlated agreement: agreement between independent reviews of the same paper relative to the expected agreement between reviews of different papers.

We next contribute an automated pipeline for learning and implementing checklists using LLM agents. Our approach takes as input a historical review corpus, and decomposes each review into atomic claims, each labeled by the source of information that would resolve it. We require the claims to be settled by the paper or its supplementary material, such that an LLM can verify the learned checklist with high accuracy. The claims are then grouped across papers into a hierarchy of named checklist items, and our proxy objective selects the granularity at which that hierarchy is applied.

We evaluate whether checklists learned using our correlated agreement objective reduce finite-reviewer error on existing review datasets. In a simulated review experiment with LLMs on 30 papers each from NeurIPS~2021 and ICLR~2022, we find that for a typical committee of three reviewers, the learned checklist reduces the review error by 58\% on NeurIPS and 21\% on ICLR. On NeurIPS 2021, it also has lower error than the official NeurIPS~2021 paper checklist at every committee size. On both datasets, the checklist with the highest correlated agreement also has the lowest error for committees of three or more reviewers, so our proxy objective identifies the checklist that best balances variance reduction and bias addition.
\section{Related Work}
Human peer review has played the role of an imperfect quality control mechanism for scientific production for over a century
~\citep{mann2025ai,shah2022challenges}. Reviewer lottery effects are demonstrated in grant funding by replicating review processes and finding that outcomes depend on the draw~\citep{cole1981chance,mayo2006peering}. In machine learning, 
a reanalysis of the 2014 duplicated assignment NeurIPS experiment attributes half of the variation in reviewer quality scores to reviewer subjectivity~\citep{cortes2021inconsistency}. The 2021 NeurIPS experiment finds substantial variation in accepted papers across two non-communicating committees who reviewed the same 882 papers~\citep{beygelzimer2023has}. 
Computational interventions proposed to address lottery issues have tended to target assignment, calibration, and dishonest behavior rather than what reviewers attend to~\citep{shah2022challenges}.

AI reviewers bring opportunities and risks. LLM-generated comments can overlap human reviews at rates comparable to a second human reviewer~\citep{liang2024large}, and models can perform well on targeted, verifiable questions
~\citep{liu2023reviewergpt}. 
State-of-the-art AI review systems increasingly rely on agentic, multi-stage pipelines that 
ground their comments in quoted passages or other verifiable evidence to improve interpretability and accuracy~\citep{gao2025reviewagents,nguyen2026benchmarking,weng2026deepreviewer,fang2026proreviewer,ma2026actreview,openaireview,reviewer3,jiang2025paperreview}. Given the infeasibility of large scale human peer review studies, such tools are often tested via recall of implanted defects; others propose LLM simulations as a testbed for peer review mechanisms that can recover aggregate dynamics of human peer review feedback and outcomes~\citep{jin2024agentreview}.
Venues have also begun to experiment with AI in review, including LLM feedback on reviewers' written reviews~\citep{thakkar2025llm}, LLM reviews returned to authors~\citep{emnlp2026aireview}, and LLM review assistants in OpenReview~\citep{neurips2026aireview}. 

However, AI reviewers can be gameable~\citep{baumann2026stop}, can hallucinate~\citep{deng2026cocoreviewbench}, and attend unevenly across review dimensions, over-weighting technical validity and under-weighting novelty~\citep{shin2025mind}. 
Applying consistent summary representations also contributes bias, by attempting to reduce scientific quality to a finite set of proxies that can cause reviewers to undervalue their private knowledge~\citep{bikhchandani1992theory,stasser1985pooling}.
\section{Problem: Learning a Checklist to Reduce Lottery Variance}
\label{sec:problem}

We present a model of peer review and the problem of reviewer lottery variance. We formulate the task of learning a checklist to reduce variance from a historical peer review corpus. 

\subsection{Reviewer Lottery Variance}

One way to frame the goal of peer review is to uncover the latent true epistemic quality of a paper. However, even if a ground truth paper quality state could be said to exist, it would be unobservable.

Thus, a more realistic formalization of the goal of peer review is to recover the expected assessment of the paper taken over the eligible reviewer population. The problem is that we can typically only observe the assessments of a small subset of the ideal group of reviewers.  

We define the population review target for a given paper $i$ as the mean assessment under the distribution of eligible reviewers for that paper:
\[
\mu_{i} = \mathbb{E}_{r\sim P_i} [Y_{ir}],
\]
where $r \sim P_i$ represents a random reviewer draw from the population of eligible reviewers for paper $i$ and $Y_{ir}$ is the judgment of an unaided reviewer $r$ on paper $i$. While in practice $Y_{ir}$ could take various forms, including multi-dimensional assessments like rubric-style evaluations, for the purposes of formalizing the problem we assume that $Y_{ir}\in\mathbb{R}$ is a scalar representing a final paper score. A peer review system that targets $m$ reviewers per paper estimates $\mu_{i}$ as:
\[
\mathbf{R}_{im} = (r_1, \ldots, r_m) \sim P_{i}^{(m)}; \qquad \bar{Y}_{im} = \frac{1}{m} \sum_{j=1}^{m}Y_{ir_{j}};
\]
where $P_i^{(m)}$ is the reviewer assignment mechanism. While $\mathbb{E}[\bar{Y}_{im}]=\mu_i$, any downstream decisions that rely on aggregated reviewer assessments are subject to sampling error, where the small values of $m$ that are typical in human peer review can yield estimates far from $\mu_{i}$.

Lottery variance may arise as a result of each reviewer's expertise being  specialized relative to the range of expertise across eligible reviewers, or reviewers having limited time or cognitive resources to spend surfacing issues, or from some reviewers making errors in interpreting aspects of the paper.

\paragraph{Estimating Reviewer Variance and Incorrect Criticism on the NeurIPS 2021 Experiment}
\label{sec:neurips}
To motivate our characterization of errors in human peer review, we contribute a novel claim-level analysis of the NeurIPS 2021 dual-committee experiment. We use the overlap between issues across the two independent committees for each paper, after merging common issues, to estimate impacts of lottery variance at the claim level. Full details are in \Cref{app:error}.    

We found that independent reviewers rarely mention the same issue, and that many concrete criticisms (those that can be checked directly against the paper) are refuted by the paper materials. Specifically:

\begin{itemize}\setlength{\itemsep}{1pt}
    \item \textbf{Low coverage}: One reviewer raises about 6\% of the distinct issues that further reviewers drawn from the same pool would eventually raise (6.4\%, 95\% CI $[5.9, 6.9]$, Chao2 richness~\citep{chao1987estimating}).
    \item \textbf{Low overlap}: Only one third of reviewer claims are also raised by the other committee (34.3\%, 95\% CI $[32.8, 35.8]$), and 81\% of issues are raised by a single reviewer (\Cref{app:error}).
    \item \textbf{False criticisms}: The paper itself clearly refutes 18.2\% of all judged concrete claims (95\% CI $[16.3, 20.2]$).\footnote{We judge this with an LLM (\Cref{sec:verify}). The check uses the camera-ready PDF, which works in the other direction: material added after review can refute a criticism that held for the submission.}
\end{itemize}

These results suggest that one reviewer raises only a small share of the issues this reviewer pool would eventually surface, and that humans make mistakes in a non-trivial proportion of review claims. 


\subsection{Reducing Variance By Surfacing a Common LLM-Completed Checklist}

\begin{figure}[t]
\centering
\resizebox{\textwidth}{!}{
\input{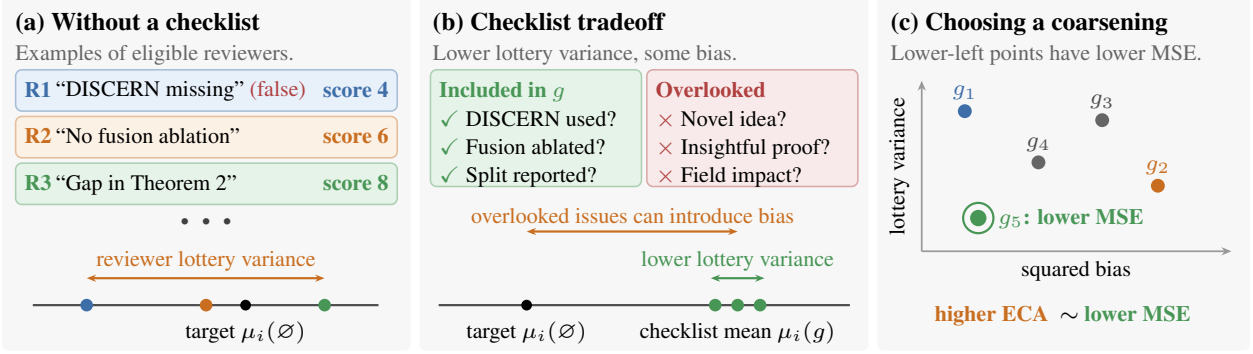}%
\tikzset{
  arr/.style={fig arrow},
  paneltitle/.style={fig title, anchor=north west},
  note/.style={fig note, anchor=north west},
  hdr/.style={fig head, anchor=north west},
  box/.style={inner sep=0, anchor=north west},
  qrow/.style={fig body, anchor=north west},
}%
\begin{tikzpicture}[x=1cm, y=1cm, font=\fontsize{7}{8}\selectfont]

\node[box, fig panel, minimum width=4.55cm, minimum height=3.95cm] (panelA) at (0.00,0) {};
\node[paneltitle] at (0.15,-0.18) {(a) Without a checklist};
\node[note] at (0.15,-0.55) {Examples of eligible reviewers.};

\node[box, fig card=stA, fill=stAl, minimum width=4.25cm, minimum height=0.46cm] (r1) at (0.15,-0.82) {};
\node[hdr, stA] at ($(r1.north west)+(0.10,-0.13)$) {R1};
\node[qrow] at ($(r1.north west)+(0.45,-0.13)$) {``DISCERN missing''\ \textcolor{stX}{(false)}};
\node[hdr, stA, anchor=north east] at ($(r1.north east)+(-0.10,-0.13)$) {score 4};

\node[box, fig card=stB, fill=stBl, minimum width=4.25cm, minimum height=0.46cm] (r2) at (0.15,-1.34) {};
\node[hdr, stB] at ($(r2.north west)+(0.10,-0.13)$) {R2};
\node[qrow] at ($(r2.north west)+(0.45,-0.13)$) {``No fusion ablation''};
\node[hdr, stB, anchor=north east] at ($(r2.north east)+(-0.10,-0.13)$) {score 6};

\node[box, fig card=stC, fill=stCl, minimum width=4.25cm, minimum height=0.46cm] (r3) at (0.15,-1.86) {};
\node[hdr, stC] at ($(r3.north west)+(0.10,-0.13)$) {R3};
\node[qrow] at ($(r3.north west)+(0.45,-0.13)$) {``Gap in Theorem 2''};
\node[hdr, stC, anchor=north east] at ($(r3.north east)+(-0.10,-0.13)$) {score 8};

\node[font=\large\bfseries, text=black!75, anchor=center, inner sep=0] at (2.28,-2.52) {\textbf{. . .}};
\node[fig body, stB, anchor=south] at (2.28,-3.07) {reviewer lottery variance};
\draw[fig span, stB] (0.95,-3.15) -- (3.60,-3.15);
\draw[line width=0.7pt, black!70] (0.35,-3.45) -- (4.20,-3.45);
\fill[stA] (0.95,-3.45) circle (2.0pt);
\fill[stB] (2.28,-3.45) circle (2.0pt);
\fill[stC] (3.60,-3.45) circle (2.0pt);
\fill[black] (2.72,-3.45) circle (1.7pt);
\node[fig body, anchor=north] at (2.72,-3.62) {target $\mu_i(\varnothing)$};

\node[box, fig panel, minimum width=5.00cm, minimum height=3.95cm] (panelB) at (4.65,0) {};
\node[paneltitle] at (4.80,-0.18) {(b) Checklist tradeoff};
\node[note] at (4.80,-0.55) {Lower lottery variance, some bias.};

\node[box, fig card=stC, fill=stCl, minimum width=2.35cm, minimum height=1.30cm] (included) at (4.77,-0.82) {};
\node[hdr, stC] at ($(included.north west)+(0.10,-0.14)$) {Included in $g$};
\node[qrow] at ($(included.north west)+(0.10,-0.46)$) {\textcolor{stC}{\checkmark}\ DISCERN used?};
\node[qrow] at ($(included.north west)+(0.10,-0.77)$) {\textcolor{stC}{\checkmark}\ Fusion ablated?};
\node[qrow] at ($(included.north west)+(0.10,-1.08)$) {\textcolor{stC}{\checkmark}\ Split reported?};

\node[box, fig card=stX, fill=stXl, minimum width=2.35cm, minimum height=1.30cm] (overlooked) at (7.18,-0.82) {};
\node[hdr, stX] at ($(overlooked.north west)+(0.10,-0.14)$) {Overlooked};
\node[qrow] at ($(overlooked.north west)+(0.10,-0.46)$) {\textcolor{stX}{$\times$}\ Novel idea?};
\node[qrow] at ($(overlooked.north west)+(0.10,-0.77)$) {\textcolor{stX}{$\times$}\ Insightful proof?};
\node[qrow] at ($(overlooked.north west)+(0.10,-1.08)$) {\textcolor{stX}{$\times$}\ Field impact?};

\node[fig body, stB, anchor=south] at (7.03,-2.55) {overlooked issues can introduce bias};
\draw[fig span, stB] (5.85,-2.63) -- (8.20,-2.63);
\node[fig body, stC, anchor=south] at (8.20,-3.07) {lower lottery variance};
\draw[fig span, stC] (7.90,-3.15) -- (8.50,-3.15);
\draw[line width=0.7pt, black!70] (4.87,-3.45) -- (9.45,-3.45);
\fill[black] (5.85,-3.45) circle (1.7pt);
\node[fig body, anchor=north] at (5.85,-3.62) {target $\mu_i(\varnothing)$};
\fill[stC] (7.95,-3.45) circle (2.0pt);
\fill[stC] (8.20,-3.45) circle (2.0pt);
\fill[stC] (8.45,-3.45) circle (2.0pt);
\node[fig body, anchor=north] at (8.20,-3.62) {checklist mean $\mu_i(g)$};

\node[box, fig panel, minimum width=4.15cm, minimum height=3.95cm] (panelC) at (9.75,0) {};
\node[paneltitle] at (9.90,-0.18) {(c) Choosing a coarsening};
\node[note] at (9.90,-0.55) {Lower-left points have lower MSE.};

\draw[arr] (10.25,-2.85) -- (13.70,-2.85);
\draw[arr] (10.25,-2.85) -- (10.25,-0.95);
\node[fig body, anchor=north] at (11.97,-2.95) {squared bias};
\node[fig body, rotate=90, anchor=south] at (10.12,-1.90) {lottery variance};

\fill[stA] (10.73,-1.29) circle (2.2pt);
\node[fig body, stA, anchor=south] at (10.73,-1.17) {$g_1$};
\fill[stB] (12.88,-2.12) circle (2.2pt);
\node[fig body, stB, anchor=south] at (12.88,-2.00) {$g_2$};
\fill[black!60] (12.26,-1.39) circle (2.2pt);
\node[fig body, text=black!70, anchor=south] at (12.26,-1.27) {$g_3$};
\fill[black!60] (11.55,-1.86) circle (2.2pt);
\node[fig body, text=black!70, anchor=south] at (11.55,-1.74) {$g_4$};
\fill[stC] (10.88,-2.48) circle (2.6pt);
\draw[stC, line width=0.7pt] (10.88,-2.48) circle (4.8pt);
\node[fig head, stC, anchor=west] at (11.10,-2.48) {$g_5$: lower MSE};

\node[fig head, anchor=north] at (11.82,-3.45) {\textcolor{stB}{higher ECA} \ $\sim$\ \textcolor{stC}{lower MSE}};
\end{tikzpicture}%
}
\caption{The problem and our selection strategy. (a)~The ellipsis denotes additional reviewers from the eligible-reviewer population. Their differing scores create reviewer lottery variance around $\mu_i(\varnothing)$, the unaided population review target. (b)~A shared checklist aligns attention on represented dimensions, lowering lottery variance, but omitted dimensions can shift the checklist-assisted population mean $\mu_i(g)$ and introduce bias. (c)~Candidate coarsenings trade lottery variance against squared bias. Because historical reviews do not reveal counterfactual MSE, we select using expected correlated agreement (ECA), $A_{\mathrm{same}}(g)-A_{\mathrm{diff}}(g)$, then evaluate variance and bias separately.}
\label{fig:problem}
\end{figure}

We seek to reduce reviewer lottery variance by introducing a shared checklist $g$. 
A checklist consists of a set of binary claims that can be resolved against the paper materials with high accuracy and surfaced for assigned reviewers. In doing so, it stands to reduce variance arising from sources like reviewer myopia or careless mistakes.
Introducing a checklist as intervention results in $Y_{ir}(g)$ at the individual reviewer level, with checklist-assisted population mean $\mu_{i}(g) = \mathbb{E}_{r\sim P_i}[Y_{ir}(g)]$ (where the no checklist target can similarly be written as $\mu_{i}(\varnothing) = \mathbb{E}_{r\sim P_i}[Y_{ir}(\varnothing)]$). 
Let the mean score of $m$ reviewers ($\mathbf{R}_{im}$) be
\[
\bar{Y}_{im}(g) = \frac{1}{m} \sum_{j=1}^{m}Y_{ir_{j}}(g).
\]

Using
\[
\mathrm{MSE}_i(g, m) = \mathbb{E}_{\mathbf{R}_{im}\overset{\text{i.i.d.}}{\sim}P_i}[(\bar{Y}_{im}(g)-\mu_{i}(\varnothing))^2]
\]
for the error of checklist-assisted review, and $\mathrm{MSE}_i(\varnothing, m)$ for the corresponding quantity under unaided review, a good checklist satisfies:
\[
\mathrm{MSE}_i(g, m) < \mathrm{MSE}_i(\varnothing, m).
\]

That is, the total error induced by the checklist will be less than that observed under no checklist. Intuitively, a checklist can reduce variance insofar as variation across reviewer scores results from variation in which issues reviewers notice or the possibility of careless mistakes checking a verifiable issue against the paper. 
However, by exposing all reviewers to a shared representation, introducing a checklist may also systematically shift their assessments away from the unaided population review target by biasing what information reviewers pay attention to or how they weight it.

By the standard bias-variance decomposition, the mean squared error of a checklist-assisted reviewer assignment can be written as:
\begin{equation}
\mathrm{MSE}_i(g, m) = \mathrm{Var}
(\bar{Y}_{im}(g)\mid i) + (\mu_{i}(g)-\mu_i(\varnothing))^2.
\label{eq:bias_variance}
\end{equation}
We want to minimize $\mathrm{MSE}_i(g,m)$. Because error in unaided review amounts to variance only:
\[
\mathbb{E}
[\bar{Y}_{im}(\varnothing)\mid i] = \mu_{i}(\varnothing) \quad \Rightarrow \quad \mathrm{MSE}_i(\varnothing,m) = \mathrm{Var}
(\bar{Y}_{im}(\varnothing)\mid i),
\]
introducing a checklist improves finite-sample estimation only when
\begin{equation}
\mathrm{Var}
(\bar{Y}_{im}(\varnothing)\mid i) - \mathrm{Var}
(\bar{Y}_{im}(g)\mid i) > (\mu_{i}(g)-\mu_i(\varnothing))^2;
\label{eq:inequality}
\end{equation}
that is, where the reduction in variance from using the checklist is larger than the added squared bias.
\Cref{fig:problem} sketches the lottery, the intervention, and why checklist design requires choosing a coarsening that reduces variance without introducing too much bias.

\subsection{Checklist Selection}
\label{sec:coarsening}
Existing review corpora provide a rich source for learning what kinds of concerns human reviewers perceive as important. 
But how should we identify the set of relevant checkable issues from which a checklist can draw? Simply distilling the space of atomic review claims from peer review corpora into checklist items would face multiple problems. First, issues learned directly from human reviews may be value judgments rather than properties verifiable against the paper or supplement, which is where LLM checks are most reliable. We therefore want to restrict to those issues that can be directly resolved by the paper materials. Second, if atomic review claims are framed too specifically or idiosyncratically, they may not capture reproducible signal in human review, and add bias if instituted across reviewers. Such a checklist could also be prohibitively large for human reviewers to process, so that only those that apply to a given paper are surfaced. 

If human reviews contain real signal about how the community views a paper, but any individual claim may be framed too idiosyncratically to recur across reviewers as more are added, then checklist design becomes a problem of granularity: a space of issues that is too low level may yield little recurrence as reviewers are added, while one that is too high level may yield recurrence simply because the issues are framed so broadly to apply to any papers. We therefore treat checklist design as finding a coarsened representation of the observed claim space, restricted to verifiable issues, on which reviewers of the same paper agree more than reviewers of different papers do.

Let $C_{ir} \subseteq \mathcal{C}$ be the set of atomic review claims extracted from reviewer $r$'s review of paper $i$, where $\mathcal{C}$ is the claim space observed across a historical review corpus. Let $\mathcal{C}_{\mathrm{ver}} \subseteq \mathcal{C}$ be the claims in this space that are verifiable against the paper or its supplementary material. A checklist is a map $g: \mathcal{C}_{\mathrm{ver}} \to \mathcal{S}_g$, where $\mathcal S_g$ is a partition of $\mathcal C_{\mathrm{ver}}$ into nonempty, disjoint claim clusters. 
Each signal $s\in\mathcal S_g$ is thus a subset of claims, and $g(c)=s$ means that claim $c$ belongs to signal $s$. 
Let $\mathcal G$ denote a finite family of candidate coarsenings $g$.

Ideally, we would select a checklist that minimizes MSE averaged over the population of papers.
Let $Q$ denote this distribution of papers, such that the population objective is $\mathbb E_{i\sim Q}[\mathrm{MSE}_i(g,m)]$.
However, doing so would require being able to observe $Y_{ir}(g)$ (reviewer $r$'s assessment of paper $i$ under checklist $g$) for each candidate checklist $g$. Existing review datasets provide only the assessment under the issues the reviewer noticed ($Y_{ir}(\varnothing)$). We therefore instead use a measure of expected correlated agreement defined on the review claim space to select the coarsening of $\mathcal{C}_{\mathrm{ver}}$ according to a property conducive to variance reduction: 
issues raised by independent reviewers of the same paper recur much more often than those raised by reviewers of different papers.

Given a checklist $g$, each review can be represented as a binary vector $Z_{ir}^{g} \in \{0,1\}^{\lvert \mathcal{S}_g \rvert}$, standing for the activation of the checklist signal by claims in review:
\begin{equation}
Z_{ir}^{g}(s) = \mathbf{1}\!\left[\,C_{ir} \cap s \neq \varnothing\,\right], \qquad s \in \mathcal{S}_g,
\label{eq:zdef}
\end{equation}
where $\mathbf{1}\!\left[\cdot\right]$ is the indicator function.
A signal is raised if the reviewer made at least one of the claims in its cluster. 
As claims are merged into fewer checklist signals, coarsening increases the chance that two reviews activate the same signal, even for unrelated papers.

Let $A(C,Z')=\sum_{c\in C\cap\mathcal C_{\mathrm{ver}}}Z'(g(c))$ count the verifiable claims of a review $C$ whose signal another review $Z'$ also raises, and let $\bar n$ be the expected number of verifiable claims per review. We define
\begin{align}
A_{\mathrm{same}}(g)
&=\bar n^{-1}\,\mathbb E_{i\sim Q}\!\left[
 \mathbb E_{r,r'\overset{\mathrm{i.i.d.}}{\sim}P_i}
 [A(C_{ir},Z_{ir'}^g)\mid i]\right],
\label{eq:asame}\\
A_{\mathrm{diff}}(g)
&=\bar n^{-1}\,\mathbb E_{i,j\overset{\mathrm{i.i.d.}}{\sim}Q}\!\left[
 \mathbb E_{r\sim P_i,\,r'\sim P_j}
 [A(C_{ir},Z_{jr'}^g)\mid i,j]\right].
\label{eq:adiff}
\end{align}
The two reviews share a paper in \eqref{eq:asame} and come from independent papers in \eqref{eq:adiff}, so both are shares of review claims whose signal the other review also raises. 
Denoting the \textbf{expected correlated agreement (ECA)} $\Delta(g)=A_{\mathrm{same}}(g)-A_{\mathrm{diff}}(g)$, we select
\begin{equation}
g^*\in\operatorname*{arg\,max}_{g\in\mathcal G}\Delta(g).
\label{eq:gstar}
\end{equation}
The objective rewards reproducible overlap within papers after subtracting overlap across papers.


\paragraph{Theoretical guarantee.}
We provide a result showing when our checklist is an approximation of the optimal checklist.
For a fixed number of reviewers $m$, let $R_m(g)=\mathbb E_{i\sim Q}[\mathrm{MSE}_i(g,m)]$ be population MSE and $B(\mathcal S_g)=R_m(\varnothing)-R_m(g)$ be the MSE reduction from checklist $g$.
\begin{theorem}[Informal]
\label{thm:agreement_mse}
Suppose MSE reduction is monotone and submodular in the set of surfaced signals, signals that contribute positively to ECA have positive stand-alone MSE benefit, and the signals of $g^*$ are not too redundant with one another. Then the maximizer $g^*$ in \eqref{eq:gstar} approximates the optimal checklist: $B(\mathcal S_{g^*})\ge\alpha\max_{g\in\mathcal G}B(\mathcal S_g)$ with $\alpha=(1-\kappa)/\rho$. Here $\rho\ge1$ is the largest ratio between two signals' stand-alone MSE benefits per unit of ECA, and $\kappa\in[0,1)$ bounds the average fraction of a signal's stand-alone benefit that is redundant given the other signals of $g^*$.
\end{theorem}
The approximation ratio depends on how consistently reviewers benefit from signals. 
The guarantee approaches optimality as $\rho\to1$ and $\kappa\to0$. Appendix~\ref{app:checklist_mse} gives the formal theorem and proof.



\section{An LLM-assisted Pipeline for Learning Checklists}
\label{sec:method}

We realize our approach using an LLM-assisted pipeline that learns a checklist from a historical review corpus by maximizing the expected correlated agreement measure $A_{\mathrm{same}}-A_{\mathrm{diff}}$ (\eqref{eq:gstar}). 

\begin{figure}[t]
\centering
\resizebox{\textwidth}{!}{
\input{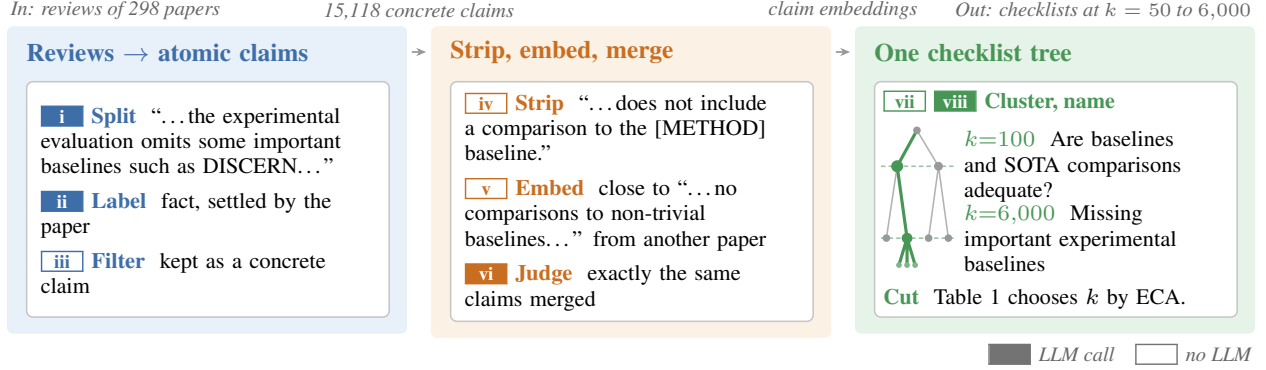}%
\newcommand{\epill}[3]{{\setlength{\fboxsep}{0.8pt}\setlength{\fboxrule}{0.45pt}%
  \ifnum#2=1\fcolorbox{#1}{#1}{\makebox[11pt]{\textcolor{white}{\fontsize{6}{6}\selectfont\bfseries\vphantom{vi}#3}}}%
  \else\fcolorbox{#1}{white}{\makebox[11pt]{\textcolor{#1}{\fontsize{6}{6}\selectfont\bfseries\vphantom{vi}#3}}}\fi}\hspace{2.5pt}}%
\begin{tikzpicture}[
  font=\fontsize{7}{8}\selectfont,
  hdr/.style={fig title, text=#1, rectangle, minimum height=0.56cm, text width=4.02cm, inner xsep=0.21cm, align=left},
  art/.style={fig annot, inner sep=1pt},
  arr/.style={fig arrow},
  stepbox/.style={rounded corners=2pt, draw=black!22, line width=0.5pt, fill=white, text width=3.72cm, minimum width=4.02cm, minimum height=2.62cm, align=left, inner xsep=4.5pt, inner ysep=3.5pt, anchor=north west, font=\fontsize{7}{8}\selectfont, execute at begin node={\hyphenpenalty=10000\exhyphenpenalty=10000\relax}},
  exlab/.style={fig body, text width=2.96cm, align=left, anchor=west, execute at begin node={\hyphenpenalty=10000\exhyphenpenalty=10000\relax}},
]
\node[hdr=stA] (A) at (0,0)    {Reviews $\rightarrow$ atomic claims};
\node[hdr=stB] (B) at (4.72,0) {Strip, embed, merge};
\node[hdr=stC] (C) at (9.44,0) {One checklist tree};
\draw[arr] ([xshift=1.5pt]A.east) -- ([xshift=-1.5pt]B.west);
\draw[arr] ([xshift=1.5pt]B.east) -- ([xshift=-1.5pt]C.west);

\node[stepbox] (eA) at ([xshift=0.21cm,yshift=-1pt]A.south west) {%
  \epill{stA}{1}{i}{\color{stA}\bfseries Split}\ \ ``\ldots the experimental evaluation omits some important baselines such as DISCERN\ldots''\\[3pt]
  \epill{stA}{1}{ii}{\color{stA}\bfseries Label}\ \ fact, settled by the paper\\[3pt]
  \epill{stA}{0}{iii}{\color{stA}\bfseries Filter}\ \ kept as a concrete claim};
\node[stepbox] (eB) at ([xshift=0.21cm,yshift=-1pt]B.south west) {%
  \epill{stB}{0}{iv}{\color{stB}\bfseries Strip}\ \ ``\ldots does not include a comparison to the [METHOD] baseline.''\\[3pt]
  \epill{stB}{0}{v}{\color{stB}\bfseries Embed}\ \ close to ``\ldots no comparisons to non-trivial baselines\ldots'' from another paper\\[3pt]
  \epill{stB}{1}{vi}{\color{stB}\bfseries Judge}\ \ exactly the same claims merged};
\node[stepbox] (eC) at ([xshift=0.21cm,yshift=-1pt]C.south west) {};

\begin{scope}[shift={(eC.north west)}]
  \node[fig body, anchor=north west] at (0.1,-0.1)
    {\epill{stC}{0}{vii}\epill{stC}{1}{viii}{\color{stC}\bfseries Cluster, name}};
  \draw[stC!55, line width=0.5pt, dash pattern=on 1.4pt off 1pt] (0.08,-0.95) -- (0.92,-0.95);
  \draw[stC!55, line width=0.5pt, dash pattern=on 1.4pt off 1pt] (0.08,-1.75) -- (0.92,-1.75);
  \draw[black!30, line width=0.5pt]
    (0.48,-0.55) -- (0.72,-0.95)
    (0.26,-0.95) -- (0.14,-1.75)
    (0.72,-0.95) -- (0.61,-1.75) (0.72,-0.95) -- (0.83,-1.75);
  \draw[stC, line width=1pt]
    (0.48,-0.55) -- (0.26,-0.95) -- (0.37,-1.75)
    (0.37,-1.75) -- (0.28,-2.05) (0.37,-1.75) -- (0.37,-2.05) (0.37,-1.75) -- (0.46,-2.05);
  \fill[black!40] (0.48,-0.55) circle (1.2pt)
    (0.72,-0.95) circle (1.4pt) (0.14,-1.75) circle (1.2pt)
    (0.61,-1.75) circle (1.2pt) (0.83,-1.75) circle (1.2pt);
  \fill[stC] (0.26,-0.95) circle (1.9pt) (0.37,-1.75) circle (1.7pt);
  \fill[stC!70] (0.28,-2.05) circle (0.9pt) (0.37,-2.05) circle (0.9pt) (0.46,-2.05) circle (0.9pt);
  \node[exlab] at (1.00,-0.95) {{\color{stC}\bfseries $k{=}100$}\ \ Are baselines and SOTA comparisons adequate?};
  \node[exlab] at (1.00,-1.75) {{\color{stC}\bfseries $k{=}6{,}000$}\ \ Missing important experimental baselines};
  \node[fig body, anchor=north west, text width=3.72cm, align=left] at (0.1,-2.3)
    {{\color{stC}\bfseries Cut}\ \ \Cref{tab:excess} chooses $k$ by ECA.};
\end{scope}

\begin{pgfonlayer}{background}
\fill[fig panel=stAl] (A.north west) rectangle ($(A.east |- eA.south)+(0,-5pt)$);
\fill[fig panel=stBl] (B.north west) rectangle ($(B.east |- eB.south)+(0,-5pt)$);
\fill[fig panel=stCl] (C.north west) rectangle ($(C.east |- eC.south)+(0,-5pt)$);
\end{pgfonlayer}

\node[art, anchor=south west] at ([yshift=1pt]A.north west) {In: reviews of 298 papers};
\node[art, anchor=south] at ($(A.north east)!0.5!(B.north west)+(0,1pt)$) {15{,}118 concrete claims};
\node[art, anchor=south] at ($(B.north east)!0.5!(C.north west)+(0,1pt)$) {claim embeddings};
\node[art, anchor=south east] at ([yshift=1pt]C.north east) {Out: checklists at $k = 50$ to $6{,}000$};
\node[art, anchor=north east] at ([yshift=-7pt]C.east |- eC.south) {\epill{black!55}{1}{}LLM call\quad\epill{black!55}{0}{}no LLM};
\end{tikzpicture}%
}
\caption{The checklist pipeline, shown on one NeurIPS~2021 review sentence. Steps (i)--(viii) follow the text; filled badges call an LLM. Counts are for the NeurIPS~2021 corpus.}
\label{fig:pipeline}
\end{figure}

Given a review corpus, we grow a tree of checklists with eight steps (\Cref{fig:pipeline}). Claims from papers for the evaluation in \Cref{sec:eval} are held out from the start, so they do not enter the tree and cannot leak into the checklist. (i)~Use an LLM (Grok~4.5) to split every review into atomic claims that refer to a self-contained, independent comment about the paper. Claims across different papers are put into a single pool. (ii)~Use an LLM to label each claim with what it asserts, using the proposition types of \citet{hua2019argument} (fact, evaluation, request, quote, reference, or non-argument), and what source would settle it: the paper, the literature, the world, or nothing (\Cref{app:extract}). (iii)~Filter any claims that are not labeled as facts or requests that the paper itself can settle. (iv)~Strip paper-specific numbers, symbols, and rare method names to make claims applicable to other papers. (v)~Embed each claim with Qwen3-Embedding-0.6B. (vi)~Merge claims judged to be exactly the same by the LLM (Grok~4.5) (\Cref{app:cluster}). (vii)~Grow a Ward tree over the merged claims, representing each by the mean embedding of its members. (viii)~Use an LLM to relabel each cluster (checklist item) with a question that can be posed against the paper, given a sample of its members.

\begin{table}[t]
\centering
\caption{Concrete-claim hierarchies for NeurIPS~2021 and ICLR~2022 at eight granularities $k$ (all values in \%). Coverage is the mean share of the $k$ items that a paper's reviewers cover. Recall is the share of 32 human-raised points (two on each of 16 papers) that an LLM finds from the paper with the checklist. ECA is $A_{\mathrm{same}}-A_{\mathrm{diff}}$ (\Cref{sec:coarsening}), the sample estimate of \eqref{eq:gstar}, with paper-level bootstrap 95\% CIs. On both corpora, $A_{\mathrm{same}}$ compares pairs of reviewers of the same paper and $A_{\mathrm{diff}}$ compares reviewers of different papers. Bold marks the highest ECA per corpus.}
\label{tab:excess}
\footnotesize
\setlength{\tabcolsep}{4pt}
\begin{tabular}{rrrrrrrrr}
\toprule
& \multicolumn{5}{c}{NeurIPS 2021} & \multicolumn{3}{c}{ICLR 2022} \\
\cmidrule(lr){2-6} \cmidrule(lr){7-9}
$k$ & Coverage & Recall & $A_{\mathrm{same}}$ & $A_{\mathrm{diff}}$ & ECA [95\% CI] & $A_{\mathrm{same}}$ & $A_{\mathrm{diff}}$ & ECA [95\% CI] \\
\midrule
50 & 42.0 & 96.9 & 22.1 & 12.4 & $+9.8$ [8.5, 11.1] & 38.9 & 27.2 & $+11.7$ [10.4, 13.1] \\
100 & 26.1 & 93.8 & 16.2 & 6.1 & $\mathbf{+10.2}$ [9.0, 11.3] & 28.8 & 15.9 & $\mathbf{+12.9}$ [11.8, 14.0] \\
200 & 14.8 & 84.4 & 12.3 & 3.2 & $+9.0$ [8.2, 10.0] & 21.7 & 9.2 & $+12.6$ [11.7, 13.5] \\
500 & 6.7 & 68.8 & 8.7 & 1.4 & $+7.3$ [6.7, 8.0] & 15.8 & 4.2 & $+11.6$ [10.9, 12.3] \\
600 & 5.7 & 65.6 & 8.2 & 1.1 & $+7.1$ [6.4, 7.8] & 14.9 & 3.5 & $+11.5$ [10.8, 12.1] \\
1{,}000 & 3.6 & 56.2 & 6.6 & 0.7 & $+5.9$ [5.4, 6.5] & 13.1 & 2.3 & $+10.8$ [10.2, 11.4] \\
2{,}000 & 1.9 & 34.4 & 4.5 & 0.4 & $+4.2$ [3.8, 4.5] & 10.9 & 1.2 & $+9.8$ [9.2, 10.4] \\
6{,}000 & 0.7 & 21.9 & 2.0 & 0.1 & $+1.8$ [1.6, 2.1] & 8.1 & 0.4 & $+7.7$ [7.2, 8.2] \\
\bottomrule
\end{tabular}
\end{table}

Given a checklist tree, cutting at any level returns a checklist at a given granularity. For example, cutting the tree near its root returns fewer, higher-level checklist items such as ``Are baselines and SOTA comparisons adequate?'' Cutting the tree near its leaves returns more lower-level checklist items such as ``Does the paper report experiments on ATIS, Advising, or WikiSQL?'' 

We next calculate the expected correlated agreement (ECA) for each checklist in the tree. Because the original human claims all lie at the leaf level, we must first map each review claim into the new checklist space. For each reviewer, we label a checklist item as applying if there is at least one original claim in that reviewer's review that falls under the checklist item. 

We demonstrate this pipeline on two review datasets: the NeurIPS 2021 dual-committee experiment~\citep{beygelzimer2023has} (the public OpenReview subset: 298 papers, each reviewed by two independent committees) and the ICLR 2022 review dataset~\citep{reviewarena2026} (300 selected papers: 175 rejected, 100 posters, 19 spotlights, and 6 orals). The NeurIPS dataset is the richest corpus because of its dual-committee design; the ICLR 2022 dataset is a more typical one. 

\Cref{tab:excess} reports ECA at eight different cuts (granularities) on both datasets. We also report coverage: for each paper, the share of the $k$ items raised by at least one of its reviewers, averaged over papers. 
We also report recall, the fraction of 32 points that human reviewers raised (two on each of 16 papers), that an LLM also finds when given the paper and the checklist (\Cref{app:recall}).
Coverage, recall, and $A_{\mathrm{diff}}$ all decrease monotonically with $k$, indicating that a checklist with a lower $k$ is more general. ECA is positive at every cut. On both venues we estimate it from pairs of reviewers of the same paper, and it is highest at $k{=}100$ ($+10.2$ points on NeurIPS and $+12.9$ on ICLR), which we take as $g^*$.
\section{Evaluation: Do the Selected Checklists Reduce Reviewer Error?}
\label{sec:eval}

We ask two questions to evaluate our checklist learning pipeline: (Q1)~when the review committee size $N$ is small, can the learned checklist reduce review MSE? (Q2)~Given a set of candidate checklists, does the proxy objective ECA choose the one that gives the minimum review MSE?

The time commitment and expertise required for peer review make a human study difficult. It could take dozens of independent reviewers to reach a stable population target $\mu_i(\varnothing)$ against which the review MSE can be computed, and reviewers would need to score each paper using multiple checklists. No existing peer review dataset satisfies either criterion, and collecting such a large review dataset is beyond the scope of this paper.

We therefore assess our method under our model assumptions with a larger LLM-simulated review study, following precedent in prior work on review tools~\citep{jin2024agentreview,erturk2026large} and evidence that LLM simulations can often replicate aggregate-level human population dynamics~\citep{ashokkumar2026large,cui2025large}. We also validate error reduction with human reviewers in a smaller experiment.

\subsection{Simulation Experiment}
\label{sec:sim_design}

Our simulation constructs a simulated reviewer population from which $\mu_i(\varnothing)$ can be estimated precisely. Each simulated reviewer mimics the limited coverage we observed in human reviews (\Cref{sec:neurips}), by focusing on a related subset of concerns and raising a small number of concrete issues before assigning a score. We then measure how the variance, bias, and total error of the committee mean change with committee size $N$, with and without a checklist. To compare the learned checklist with an existing one, we also run the same protocol with the official NeurIPS~2021 paper checklist~\citep{neurips2021checklist}, which reviewers that year were asked to use as a reference during review.

\paragraph{Datasets.}
We evaluate checklists learned on the NeurIPS~2021 dual-committee dataset~\citep{beygelzimer2023has} and the ICLR~2022 review dataset~\citep{reviewarena2026}. We evaluate on a set of 10 strong, 10 borderline, and 10 weak or rejected papers from each dataset, decided based on their review scores and decisions. Their claims are held out when the checklist tree is grown (\Cref{sec:method}), so the items are not fit on the papers used to measure error. Papers from each venue are reviewed with the checklist learned from that venue's review corpus, and all simulated reviewers are LLM agents running Composer~2.5.

\paragraph{Population target.}
We estimate the population target $\mu_i(\varnothing)$ by aggregating a large number of independent reviewers. To mimic human reviewers, we create multiple LLM reviewers with limited coverage of issues and capacity by assigning each a different branch of the checklist tree at the $k{=}50$ level (\Cref{app:eval} gives the details of the assignment). The reviewers review the paper within their branch and, in a separate round, grade it according to the paper content and their own review. We draw one such reviewer for each of the $F$ content branches ($F{=}49$ for NeurIPS and $47$ for ICLR). With this number of reviewers, the aggregate score that defines the population target $\hat\mu_i(\varnothing)$ is stable (standard error about 0.13 points; \Cref{fig:gt} in \Cref{app:results}).

\paragraph{Without checklist.}
The without-checklist setting shows how the MSE changes as the committee grows. 
A committee of size $N$ consists of $N$ reviewers drawn independently from the population described above, 
which reproduces the reviewer lottery variance observed in duplicated-review experiments~\citep{cortes2021inconsistency,beygelzimer2023has}. Since the committees are drawn from the same population that determines the target, this condition is unbiased by construction: its error consists only of variance and decreases at a rate of $1/N$. This sets the conditions to evaluate whether a checklist reduces variance by more than the squared bias it introduces (\Cref{eq:inequality}).

\paragraph{With checklist.}
In the checklist setting, each LLM reviewer reviews the paper using the learned checklist at a given granularity. The reviewer goes through every item on the checklist, keeps those applicable to the paper, and writes concrete criticisms according to the checklist items, each supported by a quote from the paper. In a separate round, it grades the paper with the same scoring prompt as the without-checklist setting. We run eight reviewers per paper, who each see the checklist in a different random order. To answer Q1, we compare the MSE of checklist-assigned committees with that of without-checklist committees of the same size. To answer Q2, we repeat the condition at several granularities, $k\in\{50,100,500,1000,2000\}$ on NeurIPS and additionally $k{=}200$ on ICLR, and check whether the granularity with the highest ECA also has the lowest MSE.

\paragraph{Metrics.}
We compare the expected variance, bias, and MSE of the committee mean with and without the checklist. For each paper and condition, we treat the observed reviewer scores as an empirical reviewer population. For a committee of size $N$, under independent sampling the variance of the committee mean is the single-reviewer variance divided by $N$. The bias is the difference between the average checklist score and the target, $\hat\mu_i(g)-\hat\mu_i(\varnothing)$, and does not depend on the committee size $N$. 
The review MSE combines the two as $\mathrm{MSE}=\mathrm{Variance}+\mathrm{Bias}^2$, following \eqref{eq:bias_variance}. We average over papers with equal weight, squaring each paper's bias before averaging so that biases of opposite sign on different papers do not cancel.

\subsection{Simulation Results}
\label{sec:sim_results}

\begin{figure}[t]
\centering
\includegraphics[width=\textwidth]{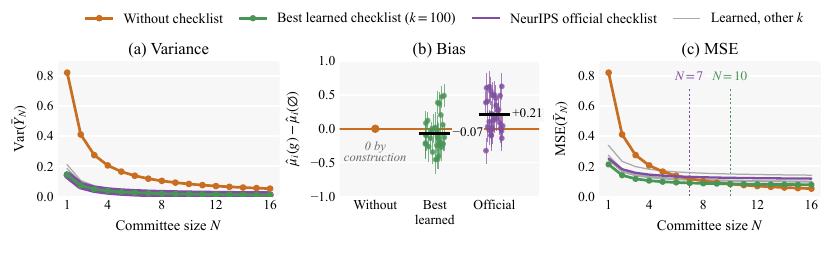}
\caption{Variance, bias, and MSE of the committee mean on NeurIPS~2021 (30 papers). (a)~Variance falls as $1/N$ in every condition but starts far lower with any checklist; the learned and official curves coincide. (b)~Each point is one paper's bias, the checklist mean minus the target $\hat\mu_i(\varnothing)$ ($\pm1$ SE); black bars average over papers. Without a checklist the bias is zero by construction. (c)~MSE, the variance plus the squared bias. Dotted lines mark the smallest $N$ at which the no-checklist committee is at least as accurate. Grey lines are learned checklists at other granularities, $k\in\{50,500,1000,2000\}$.}
\label{fig:sim}
\end{figure}

\paragraph{Q1: The learned checklist reduces MSE for small committees.}
On NeurIPS (\Cref{fig:sim}c), the best learned checklist ($k{=}100$) lowers the MSE of a single reviewer from 0.82 to 0.21 and halves it for a committee of four (0.10 vs.\ 0.21). The gain comes from lower variance, at the price of a small bias (\Cref{fig:sim}a,b). The no-checklist committee, which is unbiased by construction, catches up with the best checklist at about $N{=}10$. On ICLR (\Cref{fig:sim_iclr} in \Cref{app:results}), the $k{=}100$ checklist lowers the MSE of a single reviewer from 0.84 to 0.34, and the no-checklist committee catches up at $N{=}5$.

\paragraph{Comparison with the official checklist.}
The official NeurIPS~2021 checklist reduces variance as much as the best learned one but has higher bias. Its MSE is about 0.04 higher than that of the learned checklist at every committee size (0.25 vs.\ 0.21 at $N{=}1$, 0.14 vs.\ 0.10 at $N{=}4$).

\paragraph{Q2: ECA picks a low-MSE granularity.}
On both venues, the checklist with the highest ECA, $k{=}100$, has the lowest MSE for every committee of three or more reviewers, and the one with the lowest ECA, $k{=}2{,}000$, has the highest (at $N{=}4$, 0.10 vs.\ 0.18 on NeurIPS and 0.21 vs.\ 0.29 on ICLR), suggesting ECA is predictive of how much a checklist reduces review error. On NeurIPS, $k{=}100$ is lowest at every $N$. On ICLR, $k{=}200$ and $k{=}1{,}000$ are slightly lower at $N{=}1$ (0.32 and 0.33 vs.\ 0.34), and $k{=}200$ is level with it at $N{=}2$ (both 0.25).

\subsection{Human Study}
\label{sec:human}
In a human validation study, 20 participants each rank two sets of three NeurIPS 2021 papers, one set with the checklist at $k{=}100$ (\Cref{tab:excess}) already filled out for each paper by an LLM (Grok~4.5), and the other without a checklist. One set is on neural network pruning, which removes weights or channels to make a network cheaper to train or run, and the other is on adversarial robustness. Topic and checklist condition are counterbalanced, and \Cref{app:human_study} lists the papers. A ranking task has two advantages. The order is a less contestable target than using the small-$N$ average of human scores as $\mu_i(\varnothing)$, since the selected papers differ in mean scores (from about 5 to 7), writing quality, and citations. Ranking is also easier for participants than writing full reviews and assigning scores.

We compute variance and bias on pairs of papers, as in \citet{shivaswamy2021bias} and \citet{cormack2019bias}. Each set has three pairs. For each pair, we record 1 if a participant puts the two papers in the NeurIPS order and 0 otherwise, and apply \eqref{eq:bias_variance} with a target of 1. If a share $p$ of participants get the pair right, the variance is $p(1-p)$ and the squared bias is $(1-p)^2$. Their sum, averaged over pairs, is the fraction of pairs a participant gets wrong. With the checklist, a single ranking has variance 0.16 and squared bias 0.057 (MSE 0.22); without it, both are 0.24 (MSE 0.48). For a committee of three, the MSE is 0.11 with the checklist and 0.32 without. 

\section{Conclusion}
We formalize checklist design for peer review as selecting the coarsening that best represents the space of review claims: a shared checklist helps a small committee only if it removes more lottery variance than the bias it adds. We then build an LLM-assisted pipeline that turns a historical review corpus into a tree of paper-checkable checklist items and chooses the granularity that maximizes a proxy objective, the expected correlated agreement (ECA), which measures how much more two independent reviews of the same paper agree than reviews of different papers. We apply the pipeline to NeurIPS~2021 and ICLR~2022. In simulated review, the chosen checklist reduces the error of a three-reviewer committee by 58\% on NeurIPS and 21\% on ICLR, and on both venues the checklist with the highest ECA has the lowest error for every committee of three or more. In a human ranking study with 20 participants, the checklist lowers ranking MSE from 0.48 to 0.22.

\section*{AI Use Statement}
We used AI tools to proofread and polish the writing of this paper and to organize the code submitted with it.

\bibliography{ai_review}

\appendix
\crefalias{section}{appendix}
\crefalias{subsection}{subappendix}
\section{Full assumptions and proof for \Cref{thm:agreement_mse}}
\label{app:checklist_mse}

Let $Q$ be the target distribution of papers and $P_i$ the eligible-reviewer distribution for paper $i$. Fix the panel size $m\ge1$. Let $\mathcal C_{\mathrm{ver}}$ be a finite, nonempty space of verifiable atomic claims, and let $\mathcal G$ be a fixed, finite, nonempty family of candidate checklist mappings. Each $g\in\mathcal G$ maps a claim to its cluster in a partition $\mathcal S_g$ of $\mathcal C_{\mathrm{ver}}$: $g(c)=s$ when $c\in s\in\mathcal S_g$. Each cluster $s$ is one checklist item, also called a signal.
\[
\mathcal E=\bigcup_{g\in\mathcal G}\mathcal S_g,
\qquad
\mathcal D=\{T:T\subseteq\mathcal S_g\text{ for some }g\in\mathcal G\}.
\]
Here $\mathcal E$ contains all candidate items, and $\mathcal D$ contains all candidate checklists and their partial versions. For $T\in\mathcal D$, define
\[
\mathcal C_T=\bigcup_{s\in T}s,
\qquad g_T:\mathcal C_T\to T,
\qquad g_T(c)=s\quad\text{for the unique }s\in T\text{ containing }c.
\]
The partial mapping $g_T$ surfaces only signals in $T$ under the same rule as the original $g$. 
In particular, $g_{\mathcal S_g}=g$ and $g_{\varnothing}=\varnothing$ is the empty mapping representing unaided review.

For a mapping $h$, let $Y_{ir}(h)\in\mathbb R$ be reviewer $r$'s actual score on paper $i$. Draw $r_1,\ldots,r_m$ independently from $P_i$ and define
\[
\bar Y_{im}(h)=\frac1m\sum_{j=1}^mY_{ir_j}(h),
\qquad
\mu_i(\varnothing)=\mathbb E_{r\sim P_i}[Y_{ir}(\varnothing)\mid i].
\]
The target $\mu_i(\varnothing)$ is the unaided population mean. Define
\begin{equation}
\begin{aligned}
\mathrm{MSE}_i(h,m)&=\mathbb E[(\bar Y_{im}(h)-\mu_i(\varnothing))^2\mid i],\\
R_m(h)&=\mathbb E_{i\sim Q}[\mathrm{MSE}_i(h,m)].
\end{aligned}
\label{eq:population_risk}
\end{equation}
For fixed $m$, define the benefit and marginal benefit as set functions:
\begin{equation}
B(T)=R_m(\varnothing)-R_m(g_T),
\qquad B(s\mid T)=B(T\cup\{s\})-B(T).
\label{eq:mse_benefit}
\end{equation}
The marginal is defined for $s\notin T$ whenever $T\cup\{s\}\in\mathcal D$. Thus $B(\varnothing)=0$, and $B(\mathcal S_g)=R_m(\varnothing)-R_m(g)$.
Note that the MSE and $R_m$ take mappings such as $h$ and $g$ as arguments, while $B$ takes signal sets.

\paragraph{Signals and agreement in historical (unaided) reviews.}
Let $C_{ir}$ be the atomic claims in reviewer $r$'s unaided review of paper $i$, and let $\bar n=\mathbb E_{i\sim Q}\mathbb E_{r\sim P_i}|C_{ir}\cap\mathcal C_{\mathrm{ver}}|\in(0,\infty)$ be the expected number of verifiable claims per review. Every candidate partitions the same space $\mathcal C_{\mathrm{ver}}$, so $\bar n$ does not depend on $g$. For each signal $s\in\mathcal E$, put
\[
Z_{ir}(s)=\mathbf1[C_{ir}\cap s\ne\varnothing],
\qquad
p_s(i)=\Pr(Z_{ir}(s)=1\mid i),
\]
\[
\nu_s(i)=\frac{\mathbb E[|C_{ir}\cap s|\mid i]}{\bar n},
\qquad
d_s=\operatorname{Cov}_{i\sim Q}\bigl(\nu_s(i),p_s(i)\bigr).
\]
Here $\mathbf1[\cdot]$ is the indicator function, $p_s(i)$ is the mention probability under $P_i$, and $\nu_s(i)$ is the expected number of claims a review of paper $i$ makes in $s$, in units of the average review length. The quantity $d_s$ measures whether the papers on which reviewers write more about $s$ are also the papers on which other reviewers raise $s$.
For a candidate checklist $g$, restrict these indicators to $s\in\mathcal S_g$ to obtain the review vector $Z_{ir}^g$. As in the main text, $A_{\mathrm{same}}(g)$ is the share of review claims whose signal an independent review of the same paper also raises, and $A_{\mathrm{diff}}(g)$ is that share when the second review is of an independently drawn paper. The selection objective is $\Delta(g)=A_{\mathrm{same}}(g)-A_{\mathrm{diff}}(g)$.

\paragraph{Assumptions.}
The theorem uses the following four conditions.
\begin{enumerate}
\item[\textbf{A1.}] \textbf{Representative sampling.}
We assume historical reviews represent the target paper and reviewer populations, review pairs are sampled independently conditional on the paper.
Explicitly,
\[
\begin{aligned}
A_{\mathrm{same}}(g)
&=\frac1{\bar n}\,\mathbb E_{i\sim Q}\mathbb E_{r,r'\overset{\mathrm{i.i.d.}}{\sim}P_i}
 \left[\sum_{c\in C_{ir}\cap\mathcal C_{\mathrm{ver}}}Z_{ir'}(g(c))\,\middle|\,i\right],\\
A_{\mathrm{diff}}(g)
&=\frac1{\bar n}\,\mathbb E_{i,j\overset{\mathrm{i.i.d.}}{\sim}Q}\mathbb E_{r\sim P_i,\,r'\sim P_j}
 \left[\sum_{c\in C_{ir}\cap\mathcal C_{\mathrm{ver}}}Z_{jr'}(g(c))\,\middle|\,i,j\right].
\end{aligned}
\]
Under this condition, ECA decomposes over signals, $\Delta(g)=A_{\mathrm{same}}(g)-A_{\mathrm{diff}}(g)=\sum_{s\in\mathcal S_g}d_s$, as established in the proof.

\item[\textbf{A2.}] \textbf{Positivity.}
We assume every signal has $d_s\ge0$: papers on which reviewers write more about $s$ are not papers on which other reviewers are less likely to raise it. This holds, for example, if a review that raises $s$ makes the same expected number $\lambda_s$ of claims in $s$ on every paper, since then $\nu_s(i)=\lambda_s p_s(i)/\bar n$ and $d_s=(\lambda_s/\bar n)\operatorname{Var}_{i\sim Q}(p_s(i))$.
We further assume signals whose contribution to ECA is positive ($d_s>0$) provide positive stand-alone MSE benefits, while signals with $d_s=0$ provide zero stand-alone benefit.
Every signal with $d_s>0$ has $B(\{s\})>0$, and every signal with $d_s=0$ has $B(\{s\})=0$. 
Let $\mathcal E_+=\{s\in\mathcal E:d_s>0\}$. Define
\begin{equation}
\beta_s=\frac{B(\{s\})}{d_s}\quad(s\in\mathcal E_+),
\qquad
\rho=\max_{s,t\in\mathcal E_+}\frac{\beta_s}{\beta_t}
=\frac{\max_{s\in\mathcal E_+}\beta_s}{\min_{s\in\mathcal E_+}\beta_s}.
\label{eq:score_relevance}
\end{equation}
The quantity $\beta_s$ is the MSE benefit from signal $s$ alone per unit of ECA. The ratio $\rho\ge1$ measures the largest relative variation in this quantity across signals.

\item[\textbf{A3.}] \textbf{Monotonicity and submodularity.}
We assume adding a signal cannot increase MSE (monotonicity), and its marginal MSE reduction cannot grow as more signals are provided (submodularity). 
For $T\subseteq U$, $s\notin U$, and $U\cup\{s\}\in\mathcal D$,
\begin{equation}
B(s\mid T)\ge B(s\mid U)\ge0.
\label{eq:benefit_submodular}
\end{equation}
This condition compares compatible subsets of candidate checklists. In particular, $0\le B(s\mid T)\le B(\{s\})$, so signals that have zero stand-alone benefit have zero marginal benefit in every compatible context.


\item[\textbf{A4.}] \textbf{Limited redundancy.}
We assume signals in the selected checklist have a positive fraction of their stand-alone benefit on average when each is added after all the others.
For a candidate item set $S$ and $s\in S$, define its leave-one-out contribution
\[
\delta_s(S)=B(S)-B(S\setminus\{s\})=B(s\mid S\setminus\{s\}).
\]
For $B(\{s\})>0$, the fractional loss $1-\delta_s(S)/B(\{s\})\in[0,1]$ measures how much of signal $s$'s stand-alone benefit is redundant when all other signals in $S$ are present. For $W(S)=\sum_{s\in S}B(\{s\})>0$, define the weighted average loss
\begin{equation}
\begin{aligned}
\bar\kappa(S)
&=\sum_{s\in S\cap\mathcal E_+}\frac{B(\{s\})}{W(S)}
  \left(1-\frac{\delta_s(S)}{B(\{s\})}\right)\\
&=1-\frac{\sum_{s\in S}\delta_s(S)}{W(S)}\in[0,1].
\end{aligned}
\label{eq:average_redundancy}
\end{equation}
Zero-benefit signals carry zero weight and have $\delta_s(S)=0$ by A3. For the signal set $S^*=\mathcal S_{g^*}$ of the selected checklist $g^*$ in \eqref{eq:gstar}, we require
\begin{equation}
\bar\kappa(S^*)\le\kappa,\qquad \kappa\in[0,1).
\label{eq:limited_redundancy}
\end{equation}
\end{enumerate}
These are conditions on actual reviewer behavior.
Factors (including $\rho$ and $\kappa$) that give a useful bound in \Cref{thm:agreement_mse} must be supported by assisted-review outcomes.

The following restates the main-text result using the definitions above.
\begin{theorem}
\label{thm:agreement_mse_appendix}
Assume A1--A4. Let
\[
g^*\in\operatorname*{arg\,max}_{g\in\mathcal G}\Delta(g),
\qquad
 g_{\mathrm{opt}}\in\operatorname*{arg\,min}_{g\in\mathcal G}R_m(g).
\]
Then
\begin{equation}
B(\mathcal S_{g^*})\ge\alpha B(\mathcal S_{g_{\mathrm{opt}}}),
\qquad \alpha=\frac{1-\kappa}{\rho}\in(0,1].
\label{eq:constant_approximation_factor}
\end{equation}
Equivalently,
\begin{equation}
R_m(g^*)\le(1-\alpha)R_m(\varnothing)+\alpha R_m(g_{\mathrm{opt}}).
\label{eq:mse_guarantee}
\end{equation}
\end{theorem}

\begin{proof}
Because $\mathcal S_g$ partitions $\mathcal C_{\mathrm{ver}}$, the claim sum splits by signal: $\sum_{c\in C_{ir}\cap\mathcal C_{\mathrm{ver}}}Z(g(c))=\sum_{s\in\mathcal S_g}|C_{ir}\cap s|\,Z(s)$ for any review vector $Z$. By A1, the two reviews are independent given their papers, so the same-paper term for signal $s$ conditional on paper $i$ is $\mathbb E[|C_{ir}\cap s|\mid i]\,p_s(i)=\bar n\,\nu_s(i)p_s(i)$. When the second review is of an independently drawn paper, it is $\bar n\,\mathbb E_{i\sim Q}[\nu_s(i)]\,\mathbb E_{i\sim Q}[p_s(i)]$. Hence
\begin{equation}
\Delta(g)=\sum_{s\in\mathcal S_g}
\left\{\mathbb E_{i\sim Q}[\nu_s(i)p_s(i)]
-\mathbb E_{i\sim Q}[\nu_s(i)]\,\mathbb E_{i\sim Q}[p_s(i)]\right\}
=\sum_{s\in\mathcal S_g}d_s.
\label{eq:agreement_variance_identity}
\end{equation}
For a signal set $S=\mathcal S_g$, A3 implies, for every $s\in S$ and $T\subseteq S\setminus\{s\}$,
\[
\delta_s(S)\le B(s\mid T)\le B(\{s\}).
\]
Telescoping along any ordering of $S$ gives
\begin{equation}
\sum_{s\in S}\delta_s(S)\le B(S)\le W(S).
\label{eq:benefit_sum_bounds}
\end{equation}
Thus A4 gives the lower bound needed only for the selected checklist:
\[
B(S^*)\ge(1-\bar\kappa(S^*))W(S^*)\ge(1-\kappa)W(S^*).
\]
For the comparison across signals, let $\beta_{\min}=\min_{s\in\mathcal E_+}\beta_s>0$. This is a derived scale, not an additional assumption. By the definition of $\rho$, every signal with $d_s>0$ has $B(\{s\})=\beta_s d_s$ with $\beta_s\in[\beta_{\min},\rho\beta_{\min}]$. By A2, no signal has $d_s<0$, and signals with $d_s=0$ have zero singleton benefit. Therefore
\[
\beta_{\min}\Delta(g)\le W(\mathcal S_g)\le\rho\beta_{\min}\Delta(g).
\]
Combining this with \eqref{eq:benefit_sum_bounds} and the optimality of $g^*$ for ECA yields
\[
\begin{aligned}
B(\mathcal S_{g^*})
&\ge(1-\kappa)\beta_{\min}\Delta(g^*)\\
&\ge(1-\kappa)\beta_{\min}\Delta(g_{\mathrm{opt}})\\
&\ge\frac{1-\kappa}{\rho}B(\mathcal S_{g_{\mathrm{opt}}}).
\end{aligned}
\]
The scale $\beta_{\min}$ cancels, leaving only the relative ratio $\rho$. Substituting $B(\mathcal S_g)=R_m(\varnothing)-R_m(g)$ proves the MSE formulation.
\end{proof}

If $\mathcal E_+=\varnothing$, A2 gives zero singleton benefits and A3 forces every feasible marginal benefit to be zero. Every candidate then has $B(\mathcal S_g)=0$, and the result is trivial; one may set $\rho=1$ by convention.

\paragraph{Interpreting the approximation factor.}
The two factors in $\alpha=(1-\kappa)/\rho$ describe different possible losses. The factor $1/\rho$ accounts for how much a unit of ECA can differ in stand-alone MSE value across signals. The factor $1-\kappa$ lower-bounds the fraction of total stand-alone benefit retained by the selected checklist. 
For example, if these values differ by at most a factor $1.25$ ($\rho=1.25$) and the average diminishing value of signals is at most 20\% ($\kappa=0.2$), the selected checklist achieves at least $0.8/1.25=0.64$ of the best attainable MSE reduction. 
A4 constrains only the average loss, so individual signals in the selected checklist may still have small marginal contributions.

When $\rho=1$, every signal has the same benefit per unit of ECA. When $\kappa=0$, the selected checklist retains all its stand-alone benefit. If both hold, maximizing ECA is MSE-optimal within $\mathcal G$. 
More generally, the factor stays bounded away from zero when $\rho$ is uniformly bounded and the selected checklist's average redundancy is uniformly below one.

\paragraph{A sharper finite-size bound.}
Let $k=\max_{g\in\mathcal G}|\mathcal S_g|$. Under the same assumptions, the theorem also holds with
\begin{equation}
\alpha_k=\frac{1-\kappa+\kappa/k}{\rho}\ge\alpha.
\label{eq:limited_redundancy_factor}
\end{equation}
To see this, let $S=S^*$ and $n=|S|\le k$. For each $s\in S$, telescope along an ordering that puts $s$ first. By A3, every later signal $t$ adds at least $\delta_t(S)$, so $B(S)\ge B(\{s\})+\sum_{t\in S\setminus\{s\}}\delta_t(S)$. Averaging over the $n$ choices of $s$ gives
\[
B(S)\ge\frac{W(S)}{n}+\left(1-\frac1n\right)\sum_{t\in S}\delta_t(S)
=\left(1-\bar\kappa(S)+\frac{\bar\kappa(S)}{n}\right)W(S)
\ge\left(1-\kappa+\frac{\kappa}{k}\right)W(S).
\]
Using this bound in place of $B(S^*)\ge(1-\kappa)W(S^*)$ in the proof gives $\alpha_k$. This improves $\alpha=(1-\kappa)/\rho$ used by the main-text theorem and its restatement above.

\paragraph{Empirical selection.}
Let $\widehat\Delta$ estimate population agreement with uniform error bound
\[
\sup_{g\in\mathcal G}|\widehat\Delta(g)-\Delta(g)|\le\varepsilon,
\qquad\varepsilon\ge0.
\]
If $\widehat g$ maximizes $\widehat\Delta$, then $\Delta(\widehat g)\ge\Delta(g_{\mathrm{opt}})-2\varepsilon$. If A4 also holds with $\mathcal S_{\widehat g}$ in place of $S^*$, the same proof gives
\[
B(\mathcal S_{\widehat g})\ge\alpha B(\mathcal S_{g_{\mathrm{opt}}})
-2(1-\kappa)\beta_{\min}\varepsilon.
\]
Here $\beta_{\min}$, defined in the proof, converts agreement to squared-score units.

\paragraph{Remarks.}
The result concerns the actual benefit $B$, including reviewer biases.
Submodularity, positivity, and limited redundancy may fail under different reasoning, attention overload, or systematic misinterpretation.
Our theorem does not derive these conditions from the rationality of reviewers.

\section{Claim-level analysis of the NeurIPS 2021 reviews}
\label{app:error}

This appendix supports \Cref{sec:problem} by answering four questions about the NeurIPS~2021 dual-committee corpus: what share of a paper's issues a single reviewer raises (\Cref{sec:onerev}), how much two committees agree (\Cref{sec:overlap}), whether their criticisms are correct (\Cref{sec:verify}), and what the area chair keeps (\Cref{app:ac}). All intervals are 95\% bootstrap confidence intervals over papers.

\subsection{Data and definitions}
\label{sec:claims}

\paragraph{Corpus.}
We use the public OpenReview subset of the NeurIPS~2021 consistency experiment~\citep{beygelzimer2023has}: 298 papers, each reviewed by two committees that did not communicate, with 2{,}319 reviews and 596 area-chair (AC) decision notes in total. Papers have six to eleven reviewers, typically seven or eight.

\paragraph{Claims and layers.}
An LLM splits every review and AC note into atomic \emph{claims}. An atomic claim is an independent, self-contained comment about one problem or merit of the paper. An LLM then rewrites each claim as a statement that can be checked against the paper (\Cref{app:extract}). This gives 25{,}868 claims: 23{,}920 from reviewers and 1{,}948 from ACs. Each claim gets three labels: what kind of statement it is (fact, evaluation, request, quote, reference, or non-argument, following \citet{hua2019argument}), what would settle it (the paper, the literature, the world, or nothing), and whether it assumes the paper is missing something. The first two labels sort claims into three \emph{layers}: \emph{Concrete} claims are facts or requests that the paper itself can settle. \emph{Conceptual} claims are evaluations. \emph{Out-of-scope} claims can only be checked against other work or the world. For example, ``the experimental evaluation omits some important baselines such as DISCERN'' is concrete, and ``the techniques proposed in this paper look incremental'' is conceptual. Quotes, references, and non-arguments are dropped, because they do not assess the paper.

\paragraph{Matches and issues.}
An LLM labels pairs of claims in two settings: strict and loose (\Cref{app:pairs}); 38{,}038 pairs were labelled. A \emph{strict} match is an exact restatement, for example ``it would be helpful to have a separate discussion on the limitations of the model'' and ``the authors don't discuss the limitations of the model in any detail''. A \emph{loose} match also allows the same point at a different level of detail, for example ``only evaluated over one training dataset and one task'' and ``the paper would be strengthened by including other tasks, such as defect detection''.

Within each reviewer and layer, strictly matching claims are merged, so a point the reviewer repeats counts once. Between reviewers of the same paper, and between the two committees, claims are compared across layers, and a loose match counts as the same point; a concrete claim can therefore match a conceptual one. Loosely matching claims form an \emph{issue}, a problem that has been described from different angles in reviews. Loose matches can chain several claims together, so when a group links more than two claims, an LLM examines the whole group and splits it if its claims do not raise a single issue (\Cref{app:pairs}).

\subsection{What share of a paper's issues does a single reviewer raise?}
\label{sec:onerev}

For each paper, we merge the claims of all its reviewers into issues. Their union is the paper's observed issue set. We then estimate the richness of the pool those reviews are drawn from: how many distinct issues would appear if further reviewers were drawn in the same way. The estimator is bias-corrected Chao2~\citep{chao1987estimating,chao2005species}. It treats issues like species and reviewers like sampling occasions. We only record whether each reviewer raised an issue, so these are incidence data, the setting Chao2 is built for: if $f_1$ issues were raised by exactly one of the paper's $m$ reviewers and $f_2$ by exactly two, the estimated richness is $S_{\mathrm{obs}} + \frac{m-1}{m}\, f_1(f_1-1)/\bigl(2(f_2+1)\bigr)$. A large singleton count means many issues in this pool are still unseen. The total contains the issues this sampling process would surface. Points that reviewers in this pool would not raise are outside it.

\Cref{tab:onerev} gives the shares. A paper's reviewers raise 27 distinct issues between them, about five per reviewer. One reviewer covers a fifth of the observed set, and a committee of three or four covers a little over half. Each added reviewer keeps finding new issues: $n$ random reviewers cover 20, 35, 48, 60, 72, and 83\% of the set for $n=1$ to $6$, and 81\% of issues are raised by only one reviewer. The Chao2 estimate puts this pool at about 90 issues for a typical paper (the median; the mean of 124 is pulled up by a few papers). Of that pool, one reviewer raises about 6\%, a committee 18\%, and all of a paper's reviewers together 31\%.

\begin{table}[h]
\centering
\caption{Share of a paper's review issues raised by one reader. Observed issues are the union over the paper's reviewers. The estimated total is the Chao2 richness of issues that further draws from the same reviewer pool would raise. Means over 298 papers.}
\label{tab:onerev}
\small
\begin{tabular}{lrr}
\toprule
Reader & Of observed issues (\%) & Of Chao2 richness (\%) \\
\midrule
One reviewer & 19.7 [19.3, 20.2] & 6.4 [5.9, 6.9] \\
One committee (3 to 4 reviewers) & 57.3 [56.8, 57.7] & 18.1 [16.9, 19.2] \\
All of a paper's reviewers & 100 by construction & 30.9 [29.2, 32.6] \\
\bottomrule
\end{tabular}
\end{table}

\subsection{Do two committees raise the same issues?}
\label{sec:overlap}

A claim \emph{overlaps} if the other committee has a claim that loosely matches it, in any layer. \Cref{tab:overlap} gives the overlap rates. One reviewer claim in eight is restated exactly by the other committee, and one in three is raised at some level of detail. Judgements overlap more than specifics: the loose overlap is 40.4\% for conceptual claims and 30.9\% for concrete ones. AC claims overlap with the other committee's reviews at the same rate as reviewer claims (35.3\% versus 34.3\%). This is the lottery term of \eqref{eq:bias_variance} at the level of claims: two committees of qualified reviewers mostly do not raise the same points.

\begin{table}[h]
\centering
\caption{Cross-committee overlap: share of claims with a matching claim in the other committee. The denominator is all claims; a match the LLM misses counts as no overlap.}
\label{tab:overlap}
\small
\begin{tabular}{lrrr}
\toprule
Claims compared & $n$ & Strict (\%) & Loose (\%) \\
\midrule
Reviewer claims vs.\ other committee's reviews & 23{,}920 & 11.9 [11.1, 12.7] & 34.3 [32.8, 35.8] \\
AC claims vs.\ other committee's reviews & 1{,}948 & 11.7 [10.1, 13.3] & 35.3 [32.4, 38.2] \\
\bottomrule
\end{tabular}
\end{table}

\subsection{Are the criticisms right?}
\label{sec:verify}

For 101 papers, Grok~4.5 checked every checkable claim against the manuscript (\Cref{app:verify}). The check asks whether the paper itself bears a claim out, not whether the claim is important or fair. A claim is \emph{refuted} only when the checker finds evidence against it in the paper: it is asked to quote at least eight consecutive words from the paper, and we check automatically that the quote appears there, downgrading the verdict if it does not. Any claim can be refuted this way. A reviewer may say the paper omits something that it in fact contains, or may misstate what the paper reports or shows; in both cases the verdict rests on a passage of the paper. A typical refutation: a reviewer wrote ``I was surprised to see no mention of de Haan et al., Causal confusion in imitation learning'', and the paper says ``de Haan et al.\ (2019) explore the case when expert and imitator can observe the same contexts, but the causal diagram is not available.'' A claim is \emph{partial} when the paper has related material that does not fully settle it, \emph{unsettleable} when the paper's text cannot settle it, and \emph{unclear} when the checker cannot tell. Claims that only the literature or the world can settle could also be checked against web sources; 37 refutations, none of them concrete, rest on such a source.

The opposite verdict is narrower. A claim is \emph{upheld} when it says the paper lacks something and the checker searched for it under at least three terms without finding it (\texttt{criticism\_valid} in \Cref{app:verify}). A search can confirm that something is absent, but not that a table entry is wrong or that a proof step does not follow, so a correct criticism of that kind ends up among the partial or unclear verdicts. The upheld rate therefore undercounts correct criticisms, while every claim can be refuted. We report refuted claims as a share of all judged claims. This rate is conservative: a claim counts as wrong only when the checker produces a quote against it, and partial, unsettleable, and unclear verdicts are not counted as wrong.

The manuscript in this check is the camera-ready PDF linked from the paper's public OpenReview forum, not the anonymized submission the reviewers read. NeurIPS 2021 did not accept a revised PDF during review, but authors of an accepted paper could revise the manuscript for the camera-ready deadline~\citep{neurips2021faq}. Our corpus is the public subset of accepted duplicated papers~\citep{beygelzimer2023has}, and the submitted PDFs are not in that release. A criticism that the submission omitted a citation, baseline, or experiment can therefore be marked refuted because the authors added the missing material afterward. Of the 1{,}164 refuted claims, 89.1\% assume the paper is missing something, the kind of claim most exposed to this. The refuted rates below can be higher than the rate at which the submission itself contradicted the reviewer. This works against the conservative counting rule above, and without the submitted PDFs we cannot tell which effect is larger. The de Haan sentence above is quoted from the camera-ready PDF, and we cannot tell whether it was already in the submission.

\Cref{tab:verify} gives the verdicts by layer. Concrete claims are the ones the checker can really decide, and the paper refutes 18.2\% of all judged concrete claims. Counting only clear verdicts (upheld or refuted), reviewers and ACs are refuted at the same rate: 34.2\% and 34.3\% across all layers. Agreement does not help: claims that the other committee also raised are upheld and refuted at the same rates as claims it did not raise, within 0.8 points and with overlapping intervals.

To check the checker, we took 25 sentences from a paper $A$, turned each into a criticism that $A$ omits it, and checked each criticism against $A$, where it should be refuted, and against an unrelated paper $B$, where it should be upheld. The checker matched the answer key on 48 of 50 verdicts (96\%, Wilson 95\% CI [86.5, 98.9]).

\begin{table}[h]
\centering
\caption{Manuscript check of criticisms (101 papers). Upheld and refuted are shares of the judged claims in each layer; the rest are partial, unsettleable, or unclear.}
\label{tab:verify}
\small
\begin{tabular}{lrrr}
\toprule
Layer & $n$ & Upheld (\%) & Refuted (\%) \\
\midrule
Concrete & 5{,}527 & 34.1 [31.5, 36.8] & 18.2 [16.3, 20.2] \\
Conceptual & 1{,}776 & 13.9 [11.4, 16.9] & 6.8 [5.3, 8.4] \\
Out of scope & 444 & 24.3 [19.2, 29.6] & 8.8 [5.5, 12.2] \\
All judged & 7{,}747 & 28.9 [26.6, 31.4] & 15.0 [13.5, 16.6] \\
\bottomrule
\end{tabular}
\end{table}

\subsection{What does the area chair keep?}
\label{app:ac}

Each AC read only its own committee's reviews and then wrote the decision note, so we can see what survives from the reviews to the note. We look in both directions. Forward, an AC claim is \emph{traced} if a review the AC read contains the same issue. Backward, a reviewer issue is \emph{kept} if the AC note contains it. Among the 482 committees whose note contains at least one claim, a note has four claims on average and the reviewers raise 18 issues.

\paragraph{Where AC claims come from.}
\Cref{tab:acsrc} gives the forward direction. Just over half of AC claims can be traced to a review (54.6\% loose, 22.9\% strict). A traced claim usually restates a single reviewer: 59\% match one reviewer, 29\% two, and 12\% three or more. Among the AC claims the manuscript check could judge, 16\% of traced claims are refuted, against 23\% of the AC's own additions, with overlapping intervals.

\paragraph{What reaches the note.}
Whether an issue reaches the note depends on how many reviewers raised it: 5.7\% [5.1, 6.4] of the 7{,}607 issues raised by one reviewer reach the note, 27.1\% [24.0, 30.6] of the 783 raised by two, and 56.5\% [51.2, 61.7] of the 416 raised by three or more. Whether the issue is correct makes no difference. Among single-reviewer issues, upheld, partial, and refuted issues are each kept about 5\% of the time; among multi-reviewer issues, 30 to 41\% each, with overlapping intervals. The AC filters by consensus, not by truth.

\begin{table}[h]
\centering
\caption{Where AC claims come from: share of AC claims with a matching claim in the reviews of the AC's own committee.}
\label{tab:acsrc}
\small
\begin{tabular}{lrrr}
\toprule
AC claims & $n$ & Traced, strict (\%) & Traced, loose (\%) \\
\midrule
All & 1{,}948 & 22.9 [20.3, 25.5] & 54.6 [51.1, 57.7] \\
Concrete & 1{,}108 & 21.9 & 54.9 [50.4, 59.1] \\
Conceptual & 550 & 31.8 & 58.2 [53.9, 62.8] \\
Out of scope & 290 & 9.7 & 46.6 [40.4, 53.8] \\
\bottomrule
\end{tabular}
\end{table}

\section{Pipeline details}
\label{app:pipeline}

This appendix gives the prompts and settings behind the pipeline in \Cref{sec:method} and the scoring study in \Cref{sec:eval}; further results of the scoring study are in \Cref{app:results}. All LLM steps use Grok~4.5 through Cursor, except the simulated reviewers in \Cref{sec:eval}, which use Composer~2.5. All random seeds are 42.

\subsection{Claim extraction}
\label{app:extract}

Extraction has three passes. The first cuts a review into spans, the second rewrites each span as a statement about the paper that can be checked, and the third labels each span.

\paragraph{Pass 1: segmentation.}
We split each review field at line breaks and give the model one block at a time. The model must copy each span word for word, and any span that does not appear in the block is thrown away.

\begin{promptbox}
\begin{verbatim}
Split this block of a peer review into propositions. Do not judge them.

A proposition is the smallest span that makes one point about one topic.

- Copy each span verbatim. Do not fix typos or paraphrase.
- Keep spans in reading order, without overlap.
- Cover the whole block except greetings and section headers.
- Do not filter or cap the number of spans. Keep praise, questions, and
  typo reports too.
- Return an empty list if the block has no propositions.
\end{verbatim}
\end{promptbox}

\paragraph{Pass 2: propositionization.}
The second pass rewrites each span as a statement that a reader could mark true or false against the paper, which makes claims easier to understand and to verify. For example, ``The authors should add an ablation of the KD loss'' becomes ``the paper reports an ablation of the knowledge-distillation loss''. A pure judgement such as ``the discussion is too vague'' has no such statement and becomes null.

\begin{promptbox}
\begin{verbatim}
Rewrite each proposition as one statement about the paper that a reader could
mark true or false, or null. Point it at the paper, not at the authors.

  "The authors should add an ablation of the KD loss"
  -> "The paper reports an ablation of the knowledge-distillation loss."
  "How did you choose these learning rates?"
  -> "The paper explains how the learning rates were chosen."
  "The discussion is too vague"
  -> null

Almost every request or question has a checkable core ("the paper contains /
reports / explains X"). Use null only for a pure judgement.
\end{verbatim}
\end{promptbox}

\paragraph{Pass 3: typing.}
The third pass gives each span three labels: what kind of statement it is (a fact, an evaluation, a request, or a quote, reference, or non-argument), what would settle it (the paper, other literature, the world, or nothing), and whether it assumes the paper is missing something. The first two labels decide the layer used in \Cref{sec:claims}. Evaluations are conceptual. Facts and requests that the paper can settle are concrete. Facts and requests that only other literature or the world can settle are out of scope. Quotes, references, and non-arguments are dropped.

\begin{promptbox}
\begin{verbatim}
proposition_type: fact | evaluation | request | quote | reference | non_arg.
  fact: a checkable statement about the paper or the world.
  evaluation: a judgement ("not novel enough", "reads well").
  request: asks for an action or asks a question.
  If a span is both fact and evaluation, choose fact.

presupposes_absence: true if the span assumes the paper lacks something.
  "Please also evaluate with AutoAttack"      -> true
  "Why is SAM a second-order method?"         -> false
  "Table 4 shows accuracy below the baseline" -> false

truthmaker: what settles it. paper | literature | world | none.
  "The paper does not cite Gulcehre et al. 2021" is "paper": the citation
  list is in the paper.

Label what the words say, not what the reviewer might have meant.
\end{verbatim}
\end{promptbox}

\subsection{Same-issue pair labelling}
\label{app:pairs}

{\sloppy Labels are \texttt{exactly\_the\_same}, \texttt{a\_more\_specific}, \texttt{b\_more\_specific}, or \texttt{partially\_related}. Loose overlap is any of the first three. Pairs the model does not write are unrelated.\par}

The claim-level analysis in \Cref{app:error} uses these labels in two ways. Within each reviewer and layer, only \texttt{exactly\_the\_same} pairs are merged. Between reviewers, and between the two committees, any loose overlap counts as the same point, including across layers.

\begin{promptbox}
\begin{verbatim}
Two jobs. Write only pairs that overlap. Judge each pair on its own.

Labels: exactly_the_same | a_more_specific | b_more_specific |
        partially_related.
Loose overlap is any of the first three.

1. Within one reviewer and one layer, judge that reviewer's pairs.
   Merge only exactly_the_same.
2. Between reviewers, and between the two committees, judge every pair
   on the paper, including across layers. Loose overlap is the same point.

Sharing a topic is not enough. Different defects of the same kind (both
about writing) are at most partially_related.
\end{verbatim}
\end{promptbox}

When loose matches link more than two claims from different reviewers into one group, an LLM judges the whole group with the prompt below and splits it when its claims do not raise a single issue.

\begin{promptbox}
\begin{verbatim}
You see a group of review claims about one paper. Loose matches between
pairs of claims linked them into this group. Decide whether every claim
in the group raises the same issue about the paper.

A claim may state the issue broadly or in more detail. That is fine.
Sharing a topic is not enough: two different defects of the same kind
(both about baselines, or both about writing) are different issues.

Output:
- same_issue: true if all claims raise one issue, false otherwise.
- If false, split the claims into subgroups that each raise one issue.
  Every claim goes into exactly one subgroup. If a broad claim fits more
  than one subgroup, put it with the subgroup it fits best.
\end{verbatim}
\end{promptbox}

\subsection{Manuscript check}
\label{app:verify}

The checker receives three things: the reviewer's original sentence, its Pass~2 statement, and the full manuscript. The manuscript is the camera-ready PDF on the paper's OpenReview forum. \Cref{sec:verify} explains why checking that version can inflate the refuted rate. It may call a criticism refuted only if it quotes at least eight consecutive words from the paper as evidence. We check the quote automatically and downgrade the verdict if the quote is not in the paper. It may call a criticism valid only if it lists at least three terms it searched for.

\begin{promptbox}
\begin{verbatim}
Judge the criticism. Use the statement to know what to look for: "you never
explain how T is chosen" becomes "The paper explains how T is chosen".

Verdicts:
- criticism_refuted: the paper contains what the reviewer said was missing,
  or contradicts the reviewer.
- criticism_partial: the paper has something related but not a full answer.
- criticism_valid: you searched and the paper does not contain it.
- out_of_scope: the paper's text cannot settle this.
- unclear: you could not tell.

Rules:
1. refuted requires a verbatim quote of at least 8 words from the paper.
   The quote is checked automatically.
2. valid requires at least three search terms (synonyms, abbreviations,
   section names).
3. When unsure, say unclear. Never guess refuted.
4. A paper that explains something badly still explains it: partial,
   not refuted.
5. Ignore whether the criticism is important or fair. Only whether the paper
   bears it out.
\end{verbatim}
\end{promptbox}

The checker also records whether the span is a criticism, a description of the paper, or a question. A description that turns out to be wrong is not counted as a wrong criticism. \texttt{criticism\_valid} only records that a claimed absence is real. A positive error that the paper does contain, such as a wrong table entry or a proof step that does not follow, receives no upheld label (\Cref{sec:verify}).

\subsection{Building the hierarchy}
\label{app:cluster}

\paragraph{Fingerprint stripping.}
From here on, each claim's text is its Pass~2 statement, or the original span when Pass~2 returned null. Reviewers of different papers describe the same issue with different line numbers and method names, so before embedding we replace these with placeholders: numbered references become \texttt{[LINE]}, \texttt{[SEC]}, \texttt{[FIG]}, \texttt{[TAB]}, \texttt{[EQ]}, \texttt{[ALG]}, or \texttt{[PAGE]}, and rare method names become \texttt{[METHOD]}.

\paragraph{Embedding.}
We embed each text with Qwen3-Embedding-0.6B using the instruction below, then reduce it to 128 dimensions with a random projection and normalize it.

\begin{promptbox}
\begin{verbatim}
Instruct: Represent the underlying issue raised by this scientific peer-review
claim. Claims about the same issue should be close even across different papers.
Preserve distinctions between a broad issue and a more specific version of it,
and ignore paper-specific names, numbering, and reviewer phrasing.
Query:
\end{verbatim}
\end{promptbox}

\paragraph{Tree.}
Claims from the papers used in \Cref{sec:eval} are held out from the start of the pipeline, so those papers do not enter the tree. This leaves 15{,}118 concrete claims. Step~(vi) of \Cref{sec:method} pools all of them and merges claims that an LLM judges to be exactly the same. Candidate pairs link each claim to its three nearest neighbours by cosine similarity, computed on the full Qwen3 embeddings before the 128-dimensional projection, and the LLM judges each candidate pair with the prompt below. This step is separate from the claim-level analysis in \Cref{app:error}. Each connected component of merged pairs becomes one point, whose embedding is the mean of its members' embeddings. Step~(vii) grows a Ward tree on the embeddings of these points.

\begin{promptbox}
\begin{verbatim}
You see pairs of concrete review claims. The two claims in a pair can come
from different papers. Paper-specific numbers and method names are already
replaced by placeholders such as [METHOD] and [TAB].

For each pair, decide whether the two claims are exactly the same.
- exactly_the_same: the same issue at the same level of detail, so either
  claim could stand in for the other as a checklist item.
- not_same: anything else, including when one claim is broader than the
  other, or when both raise different defects of the same kind (both
  about baselines, or both about writing).

Judge each pair on its own. Sharing a topic is not enough.
\end{verbatim}
\end{promptbox}

\paragraph{Cuts.}
Every checklist ($k \in \{50, 100, 200, 500, 600, 1000, 2000, 6000\}$) is a cut of this one tree, so each coarse item is the union of the finer items under it.

\paragraph{Naming.}
At each cut, an LLM writes a short question for each cluster, given the cluster's size, the number of papers it draws on, its parent, and a sample of its claims. These questions are the checklist items.

\subsection{Recall against verified issues}
\label{app:recall}

The recall column in \Cref{tab:excess} asks whether an LLM given the checklist at one cut finds the issues that human reviewers found. The targets are 32 human criticisms, two from each of 16 papers, that the manuscript check upheld. Each cut is a separate run. The LLM receives the paper and that cut's checklist, the $k$ questions, and judges each question on its own. A target is found when the LLM flags a gap on the item that contains it. Recall at cut $k$ is the share of the 32 targets found.

\begin{promptbox}
\begin{verbatim}
Judge every question on this checklist on its own.

For each question, output gap_present:
- true: the paper is missing, weak, wrong, or unclear on this issue.
- false: the paper addresses it, or it does not apply.
Partial coverage is a gap if a substantive part is missing.

Do not read the reviews. Do not use the examples as evidence about this paper.
\end{verbatim}
\end{promptbox}

\subsection{The scoring study}
\label{app:eval}

Each simulated reviewer works in two calls: it first writes its criticisms, then reads them back and gives a score. Splitting the two keeps the model from picking a score first and writing a review to match.

\paragraph{Without checklist.}
Each reviewer is assigned one branch of the $k{=}50$ cut and ten random $k{=}6{,}000$ items from that branch, and may add up to two issues of its own within the branch. Because the branches do not overlap, neither do the items that different reviewers receive, so each reviewer has its own main focus. The simulation models the first round of review, so branches about rebuttal and revision are skipped (one on NeurIPS, three on ICLR), which leaves one reviewer for each of the $F{=}49$ remaining branches on NeurIPS and $F{=}47$ on ICLR. The average score over all $F$ reviewers is the population target $\hat\mu_i(\varnothing)$, and \Cref{app:gt} shows how closely committees of different sizes approach it.

\begin{promptbox}
\begin{verbatim}
You are a reviewer assigned one k=50 field and 10 random k=6000 claims under it.
Judge only those claims. You may add up to 2 more concrete issues from the
same field.

- Read the whole paper. Stay inside your field.
- Do not read other reviewers or any labels.
- Do not turn a claim into a generic question ("is the paper novel?"). If a
  claim names a method or dataset this paper never uses, answer na.
- verdict: gap (missing, weak, or wrong), ok (handled), or na (does not apply).
- Quote at least one full sentence from the paper for every gap or ok.
\end{verbatim}
\end{promptbox}

\paragraph{With checklist.}
Eight reviewers per paper receive the same checklist, each in a different random order, and may open the finer items under any item. The prompt below is the $k{=}100$ version; the other granularities use the same prompt with their own items.

\begin{promptbox}
\begin{verbatim}
Read the whole paper and the k=100 checklist. Use the finer items under each
item as a guide to what it means.

1. Go through every item. Keep only the ones plausible for this paper.
2. The finer items and other-paper examples explain the item; they are not
   evidence about this paper.
3. Search the whole paper before calling an item missing.
4. Write at most eight concrete criticisms, fewer if fewer are well supported:
   omissions, inconsistencies, unsupported claims, or inadequate experiments
   that the paper text can settle.
5. Quote at least one full sentence from the paper for each criticism.
   Items you skipped are not defects.
\end{verbatim}
\end{promptbox}

\paragraph{Estimand.}
For each paper and condition, we compute variance, bias, and MSE exactly from the observed scores (8 reviewers per checklist condition and $F$ unaided reviewers) and average them over the 30 papers of each venue with equal weight, squaring each paper's bias before averaging.

\paragraph{Scoring.}
Both conditions use the same scoring prompt, apart from the phrase that names the condition. The checklist version:

\begin{promptbox}
\begin{verbatim}
You already wrote criticisms for this paper. Now give your NeurIPS score (1-10).
Do not write a new review.

- Read the paper and your own review. Do not read other reviewers.
- Score the paper as you would after reading it with that focus. Gaps you
  flagged pull the score down; items you skipped are not defects.
- Give a score only, not an accept/reject decision.

Scale: 10 top 5 percent; 8 accept; 6 marginally above; 5 marginally below;
3 clear reject; 1 trivial or wrong. 7 and 9 are allowed.
\end{verbatim}
\end{promptbox}

\subsection{Human study papers}
\label{app:human_study}

Both sets in \Cref{sec:human} come from NeurIPS~2021. Within each set, papers are listed in order of their NeurIPS mean rating. The pruning papers are GraNet, ``Sparse Training via Boosting Pruning Plasticity with Neuroregeneration'' (7.0); CHIP, ``CHIP: CHannel Independence-based Pruning for Compact Neural Networks'' (6.0); and ``Efficient Neural Network Training via Forward and Backward Propagation Sparsification'' (5.4). The adversarial robustness papers are NuAT, ``Towards Efficient and Effective Adversarial Training'' (7.0); ``Data Augmentation Can Improve Robustness'' (6.0); and ``Shift Invariance Can Reduce Adversarial Robustness'' (5.25). Like the simulation papers, all six are held out when the checklist tree is grown (\Cref{sec:method}).

\section{Additional results of the scoring study}
\label{app:results}

This appendix gives further and more detailed results of the scoring study in \Cref{sec:eval}: how stable the population target is (\Cref{app:gt}), the ICLR counterpart of \Cref{fig:sim} (\Cref{app:iclr}), and how the spread of simulated scores compares with that of human scores (\Cref{app:humans}). The prompts and settings of the study are in \Cref{app:eval}.

\subsection{Stability of the population target}
\label{app:gt}

The population target $\hat\mu_i(\varnothing)$ is the average score over all $F$ unaided simulated reviewers of a paper (\Cref{app:eval}). \Cref{fig:gt} shows how closely committees of different sizes approach it.

\begin{figure*}[t]
\centering
\includegraphics[width=\textwidth]{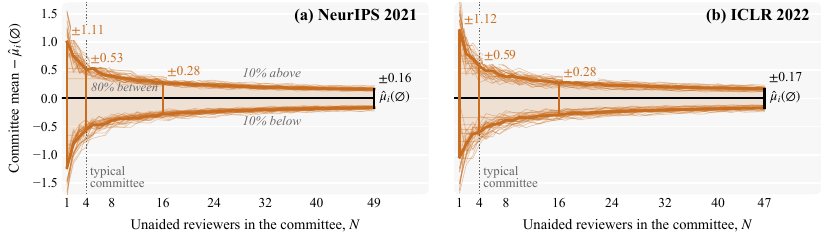}
\caption{Variance in the average review score as more reviewers are added. For each paper, we draw committees of $N$ unaided reviewers at random (with replacement) from its $F$ simulated reviewers ($F{=}49$ for NeurIPS, $47$ for ICLR), average their scores, and subtract the paper's target $\hat\mu_i(\varnothing)$, the average over all $F$. Zero means the committee matches the target. Each paper has two thin lines: a randomly drawn committee of size $N$ lands above the upper line 10\% of the time, below the lower line 10\% of the time, and between them the remaining 80\%. The thick lines and shading show the typical paper, taking the median of these lines over the 30 papers. A single reviewer is typically within about $\pm1.1$ points of the target, a committee of four within about $\pm0.55$, and all $F$ reviewers within about $\pm0.17$, which is the remaining uncertainty of $\hat\mu_i(\varnothing)$ itself.}
\label{fig:gt}
\end{figure*}

\subsection{ICLR results}
\label{app:iclr}

\Cref{fig:sim_iclr} repeats \Cref{fig:sim} for ICLR~2022. As on NeurIPS, the $k{=}100$ checklist cuts single-reviewer variance from 0.84 to 0.18 and has a small mean bias ($-0.06$). Its per-paper biases spread more widely than on NeurIPS, so its mean squared bias is 0.16 against 0.07 there, and the no-checklist committee catches up sooner, at $N{=}5$ rather than $N{=}10$. At $N{=}4$ the checklist is still marginally ahead (0.205 against 0.209).

\begin{figure}[t]
\centering
\includegraphics[width=\textwidth]{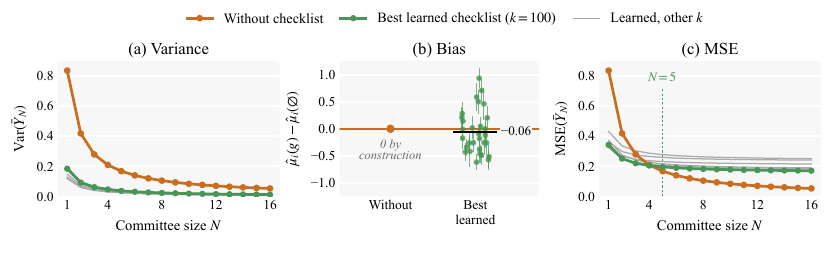}
\caption{Variance, bias, and MSE of the committee mean on ICLR~2022 (30 papers), drawn as in \Cref{fig:sim}. ICLR has no official checklist, so only learned checklists appear. (b)~Each point is one paper's bias, the checklist mean minus the target $\hat\mu_i(\varnothing)$ ($\pm1$ SE); the black bar averages over papers. (c)~The dotted line marks the smallest $N$ at which the no-checklist committee is at least as accurate as the $k{=}100$ checklist. Grey lines are learned checklists at other granularities, $k\in\{50,200,500,1000,2000\}$.}
\label{fig:sim_iclr}
\end{figure}

\subsection{Simulated and human reviewer variance}
\label{app:humans}

The unaided simulated reviewers stand in for the reviewer lottery, so their scores for a paper should spread about as much as human scores for that paper do. \Cref{tab:sim_human} and \Cref{fig:sim_human} compare the within-paper spread of single-reviewer scores on the 30 papers of each venue. Humans number three to nine per paper, so every variance here uses the unbiased estimator, with $n-1$ in the denominator.

\begin{table}[h]
\centering
\caption{Within-paper spread of single-reviewer scores on the 30 papers of each venue, for simulated unaided reviewers and for the human reviewers of the same papers. $n$ is the mean number of reviewers per paper, or per committee in the within-committee row. The last column is the share of same-paper reviewer pairs whose scores differ by at least two points. Brackets are paper-level bootstrap 95\% CIs.}
\label{tab:sim_human}
\small
\begin{tabular}{llrll}
\toprule
Venue & Reviewers & $n$ & Variance & Pairs $\ge2$ apart (\%) \\
\midrule
NeurIPS 2021 & Simulated, without checklist & 49 & 0.84 [0.77, 0.92] & 23.2 [21.0, 25.4] \\
& Human, within one committee & 3.9 & 0.93 [0.63, 1.33] & 20.7 [14.6, 27.2] \\
& Human, both committees pooled & 7.7 & 1.50 [1.18, 1.84] & 37.4 [30.9, 43.9] \\
\midrule
ICLR 2022 & Simulated, without checklist & 47 & 0.86 [0.76, 0.95] & 23.7 [20.9, 26.6] \\
& Human, one committee & 3.8 & 1.44 [0.74, 2.28] & 32.6 [21.2, 43.3] \\
\bottomrule
\end{tabular}
\end{table}

\begin{figure}[ht]
\centering
\includegraphics[width=\textwidth]{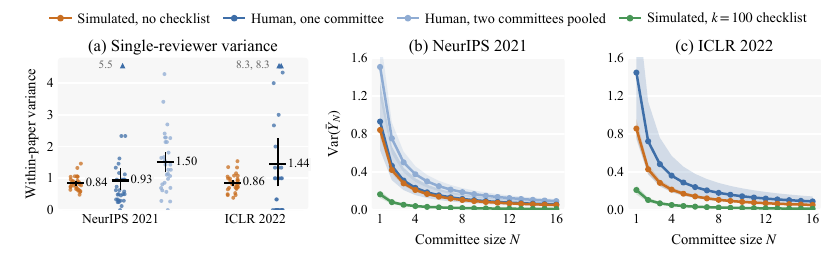}
\caption{Score spread of simulated and human reviewers on the 30 papers of each venue. (a)~Each point is one paper's within-paper variance of single-reviewer scores (unbiased estimator). Black bars are means with paper-level bootstrap 95\% CIs; triangles mark papers above the axis, labelled with their values. For NeurIPS, the one-committee value averages the variances of the paper's two committees. (b,~c)~Variance of the committee mean, the mean variances in (a) divided by $N$, drawn as in \Cref{fig:sim}a with human committees added and the $k{=}100$ checklist for scale. Shading is the 95\% CI; the ICLR human band reaches 2.28 at $N{=}1$, above the axis.}
\label{fig:sim_human}
\end{figure}

On NeurIPS, humans and simulated reviewers use the same 1-to-10 scale. The simulated variance, 0.84, is close to the variance among the reviewers of one human committee, 0.93, and so is the share of reviewer pairs two or more points apart (23\% and 21\%). Pooling the two committees raises the human variance to 1.50, because reviewers agree more with their own committee than with the other one. The same gap appears over all 298 NeurIPS pairs (1.07 within a committee, 1.33 pooled). The simulation draws reviewers independently and has no committees, so it matches the spread within a human committee but not the extra spread between committees. In \Cref{fig:sim_human}b, the simulated curve stays inside the one-committee band at every $N$, and the gap between the two is small next to the drop to the $k{=}100$ checklist curve.

On ICLR, each paper has one committee of about four reviewers, who score on the ICLR~2022 scale of 1, 3, 5, 6, 8, and 10. On that scale any two different scores other than 5 and 6 are at least two points apart, while the simulated reviewers use the 1-to-10 scale of the scoring prompt in \Cref{app:eval}. As a result, 62\% of human pairs give the same score and only 5\% differ by one point, against 32\% and 44\% of simulated pairs. The human variance, 1.44, is higher than the simulated 0.86, but with four reviewers per paper its interval is wide, and the ratio of the two, 1.69 [0.87, 2.70], includes one. In \Cref{fig:sim_human}c, the simulated curve stays inside the human band at every $N$.

The simulation does not reproduce which papers split their reviewers: per-paper human and simulated variances are uncorrelated (Spearman $-0.27$ on NeurIPS, $0.05$ on ICLR). Human committees do not agree on this either. Over all 298 NeurIPS pairs, the per-paper variances of the two committees correlate at $0.09$. Paper means agree more closely: across the 30 papers, which run from strong to weak, the simulated target $\hat\mu_i(\varnothing)$ and the mean human rating have Spearman correlation $0.79$ on NeurIPS and $0.77$ on ICLR.

For reference only, we also score the human ratings against $\hat\mu_i(\varnothing)$ as in \Cref{fig:sim}; these results are not reported in the main text. A single human rating has MSE 2.123 on NeurIPS and 2.874 on ICLR, and human ratings sit on average 0.38 points below $\hat\mu_i(\varnothing)$ on NeurIPS and 0.33 points above it on ICLR. Because $\hat\mu_i(\varnothing)$ is defined by LLM reviewers, this gap says nothing about human bias.

\end{document}